\documentclass{article}
\usepackage{graphicx}
\usepackage{amsmath, amsthm, amsfonts}
\usepackage{thmtools}
\usepackage{algorithm}
\usepackage{algpseudocode}
\usepackage{hyperref}
\usepackage{fullpage}
\usepackage[]{color-edits}
\usepackage{enumitem}
\usepackage[numbers]{natbib}

\addauthor{AK}{blue}

\addauthor{TE}{red}

\newcommand{\STA}[1]{\texttt{STA}_{#1}}
\newcommand{\ALG}{\texttt{ALG}}
\newcommand{\OPT}{\texttt{OPT}}
\newcommand{\ind}[1]{{1}_{#1}}
\newcommand{\T}{\mathbb{T}}

\newcommand{\E}{\mathbb{E}}

\newtheorem{definition}{Definition}
\newtheorem{claim}{Claim}
\newtheorem{lemma}{Lemma}

\newtheorem{example}{Example}
\newtheorem{remark}{Remark}

\title{Identity-Truthful Online Decision-Making}
\author{Tomer Ezra\thanks{Tel Aviv University, Tel Aviv, Israel. Email: \texttt{tomerezra@tauex.tau.ac.il}} \and Adar Kantor\thanks{Tel Aviv University, Tel Aviv, Israel. Email: \texttt{adar.kantor@gmail.com}}}
\date{}

\begin{document}

\maketitle

\begin{abstract}
In Bayesian online settings, a decision-maker needs to select a reward among a sequence of stochastic rewards. Upon the arrival of each reward, the decision-maker decides immediately and irrevocably whether to accept the current reward or discard it and continue to the next reward. The rewards are distributed according to known distributions, which are also referred to as their \textit{identities}. While classical prophet inequalities compare online algorithms to the offline optimum, a recent line of work studies information gaps by comparing algorithms with restricted information to the optimal online algorithm, which knows the full arrival order of the identities of the rewards. In this paper, we introduce and study a new restriction motivated by identity-independent decision-making: the decision-maker observes the current realized value, but learns the identity of the current reward only after deciding whether to accept it. We call algorithms satisfying this restriction \emph{identity-truthful}.

We define the \emph{identity-truthfulness gap} as the optimal worst-case approximation achievable by identity-truthful algorithms relative to the optimal online benchmark. 
This gap lies between the order-competitive ratio 
\citep{ezra:2023-whoisnext} and the identity-blindness gap 
\citep{ezra:2024-behindtheveil}.
We design an algorithm establishing that the identity-truthfulness gap is strictly greater than $0.5$, thereby separating it from the identity-blindness gap, which is exactly $0.5$.
We complement this result with an upper bound of $0.81$ on the identity-truthfulness gap, which is strictly lower than the known upper bound of $0.829$ on the order-competitive ratio \citep{chen:2024-settingtargets}.

Finally, we study the relationship between identity-truthful algorithms and pricing mechanisms. Motivated by questions on the power of posted prices in prophet inequalities \citep{dutting:2020}, we show that identity-truthfulness breaks the usual connection between online algorithms and pricing: there exists an instance for which the optimal identity-truthful algorithm cannot be implemented by any pricing mechanism.
\end{abstract}

\section{Introduction}\label{sec:intro}
We consider an optimal stopping problem under Bayesian information, where a decision-maker faces a sequence of $n$ stochastic rewards and must select one of them in an online fashion.
Traditionally, this problem has been analyzed through the prophet inequality framework, introduced by \citet{krengel:1977,krengel:1978} and \citet{ester:1984}. \citet{ester:1984} demonstrated that a static, anonymous, single-threshold algorithm achieves a $\frac12$-competitive ratio against the optimum offline benchmark, and that no online strategy can guarantee a better approximation.
One reason prophet inequalities have attracted significant interest is that they translate directly into pricing mechanisms with strong strategyproofness guarantees.
Specifically, \citet{chawla:2010} showed that applying a prophet inequality algorithm to the distribution of virtual values yields a (personalized) pricing mechanism that maintains the same approximation guarantee for single-parameter settings.
Moreover, \citet{correa:2019} showed that the converse also holds: finding an optimal dynamic mechanism is equivalent to solving the corresponding prophet inequality problem.

This connection to mechanism design motivated \citet{dutting:2020} to ask whether truthfulness in general, extending beyond single-parameter environments, is without loss of generality.
Subsequently, \citet{banihashem:2024} answered this question by demonstrating that every prophet inequality algorithm can be transformed into a personalized dynamic bundle pricing mechanism with matching performance guarantees.
However, truthfulness has multiple dimensions.
For instance, the personalized pricing mechanisms that they employ are not truthful with respect to revealing customer identities.
In this paper, we focus on mechanisms that satisfy this specific notion of identity independence, which we term \textit{identity-truthful}.

Our notion of identity-truthfulness connects to a growing recent literature on optimal stopping with Bayesian information, which evaluates performance relative to the optimum online benchmark rather than the offline prophet.
Comparing to the optimum online benchmark offers several advantages:
it provides a more realistic benchmark for comparison;
it serves as a tool to design novel algorithms with tighter guarantees;
it reveals separations in performance between different algorithmic models;
and it measures the impact of varying information structures available to the decision-maker.
For example, the order-competitive ratio introduced by \citet{ezra:2023-whoisnext} revealed that adaptive algorithms are necessary to obtain tight approximations against the optimum online benchmark, contrasting with the prophet benchmark for which \citet{ester:1984} proved that a single static threshold suffices.
The order-competitive ratio evaluates how well an order-unaware algorithm, that does not know the future arrival order, performs against an optimum online benchmark that knows the arrival order from the outset.
They proved that a deterministic algorithm guarantees a $\frac{1}{\varphi}$ fraction of the optimum online value, and that this bound is tight (for deterministic algorithms).
Subsequently, \citet{chen:2024-settingtargets} presented a simpler algorithm that achieves the same $\frac{1}{\varphi}$ ratio, and an improved randomized online algorithm, while establishing an upper bound of $0.829$ for all (randomized) algorithms. 
Conversely, \citet{ezra:2024-behindtheveil} considered the identity-blindness gap, where the decision-maker knows the prior distributions of all rewards but cannot see the identity of any arriving item, thereby keeping the arrival order completely hidden throughout the decision-making process.
They showed that the identity-blindness gap against the optimum online algorithm is exactly $\frac{1}{2}$, demonstrating that hiding identities completely eliminates the performance advantages of comparing to the optimum online benchmark.

In this paper, we study a distinct information gap by analyzing the performance of online algorithms that can access the identities of previously observed rewards and the value of the current reward, but must make an accept-or-reject decision without knowing the current reward's identity.
We define this class of algorithms as \textit{identity-truthful} algorithms. This informational regime sits directly between the identity-blindness gap and the order-competitive ratio.
Because an identity-truthful algorithm cannot access the identity of current or future boxes, it is inherently order-unaware.
Yet, unlike completely identity-blind strategies, it can utilize historical identity data to inform future decisions.
Notably, knowing the identity of the arriving box is at the core of the algorithms that \citet{ezra:2023-whoisnext} and \citet{chen:2024-settingtargets} designed\footnote{In Appendix~\ref{apx:order-unaware}, we show that the immediate adaptations of these algorithms (which exclude the identity of the arriving box) fail to yield better than $0.5$ approximation compared to the optimum online algorithm.}. 
Thus, the resulting identity-truthfulness gap must lie between $\frac{1}{2}$ and $\frac{1}{\varphi}$ for deterministic algorithms and between $\frac{1}{2}$ and $0.829$ for randomized algorithms.

Conceptually, the identity-truthfulness gap introduces a measure of fairness or non-discrimination, ensuring that the decision-maker evaluates each remaining option based solely on its realized value rather than its identity. 
For example, a retail store might update its baseline prices before each customer arrives and adjust them dynamically after a customer departs, yet the store cannot personalize the price based on the specific identity of the current customer upon arrival. 
Our work formalizes and analyzes this identity-truthful gap. We now present a numerical example that illustrates the identity-truthfulness gap.

\begin{example}[identity-truthfulness gap]
    Consider an instance with two boxes. Box A has a reward of $1$ with probability $\frac{1}{2}$ and $0$ otherwise, and Box B has a reward of $\frac{1}{\epsilon}$ with probability $\epsilon$ and $1$ otherwise.
        
    Box B has expectation $2-\epsilon$ and its value is always greater than or equal to Box A across all realizations. Therefore, the optimum online algorithm, which knows the arrival time of this box, always chooses it. The expectation of this algorithm approaches $2$ (as $\epsilon$ goes to $0$).
    On the other hand, an identity-truthful algorithm, upon observing a value of $1$, cannot ascertain from which box it was drawn. It is always optimal to accept $\frac{1}{\epsilon}$ and reject $0$, but it is not clear what to do upon observing $1$ on the first arriving box.
    
    If the algorithm accepts a value of $1$ on the first box, then under the worst-case order (where Box A arrives first), the algorithm achieves in expectation $\frac{1}{2}\cdot1+\frac{1}{2}\cdot(2-\epsilon)\to\frac{3}{2}$. On the other hand, if the algorithm rejects a value of $1$ on the first box, then under the worst-case order (where Box B arrives first), the algorithm achieves in expectation $\epsilon\cdot\frac{1}{\epsilon}+(1-\epsilon)\cdot\frac{1}{2}\to\frac{3}{2}$.
    Therefore, the identity-truthfulness gap for this instance (for deterministic algorithms) is $\frac{3}{4}$.
\end{example}

\subsection{Our Contributions \& Techniques}

Our main result in this paper is a separation between the identity-truthfulness gap and the identity-blindness gap (which is known to be $\frac{1}{2}$ \citep{ezra:2024-behindtheveil}). 

\begin{restatable}{theorem}{maintheorem}\label{thm:Beating 0.5}
    The identity-truthfulness gap is strictly greater than $\frac12$. 
\end{restatable}

Our analysis is partitioned into three cases. We first detect ``easy'' instances that admit good guarantees through single-threshold algorithms, and instances that require dedicated algorithms (which we term ``threshold-hard'' instances).
We further distinguish between threshold-hard instances that contain a box with expectation at least $\alpha$-fraction of the prophet's expectation. We term such a box (if it exists) a major-box.

Our analysis relies on the following properties:  
If an instance is threshold-hard then it admits a ``deterministic \& free-reward'' decomposition, meaning the prophet's distribution is close (by some metric) to a reward of $1/\epsilon$ with probability $\epsilon$, and 1 otherwise (after normalization).
Moreover, we can define the residual prophet and the residual free-reward at each step, which play a key role in the design of our algorithm. 

Our algorithm can be described as follows: If the instance is not threshold-hard, then we use a single-threshold algorithm that obtains an approximation better than $0.5$ (with respect to the prophet, and thus also with respect to the optimum online).

Else, if the instance has no major-box (in addition to being threshold-hard), our algorithm starts with a high threshold to accept only free-rewards, and once the residual prophet is low enough, the threshold is reevaluated as the better of two thresholds: half of the residual prophet, and keeping the high threshold (obtaining the residual free-reward). Since in this case the instance has no major-box, we show that the residual prophet decreases with sufficient granularity to guarantee a good reevaluation step to obtain an approximation better than $0.5$. This is illustrated in Figure~\ref{fig:XZ-plane}.

\begin{figure}[ht]
    \centering
    \includegraphics[width=0.85\textwidth]{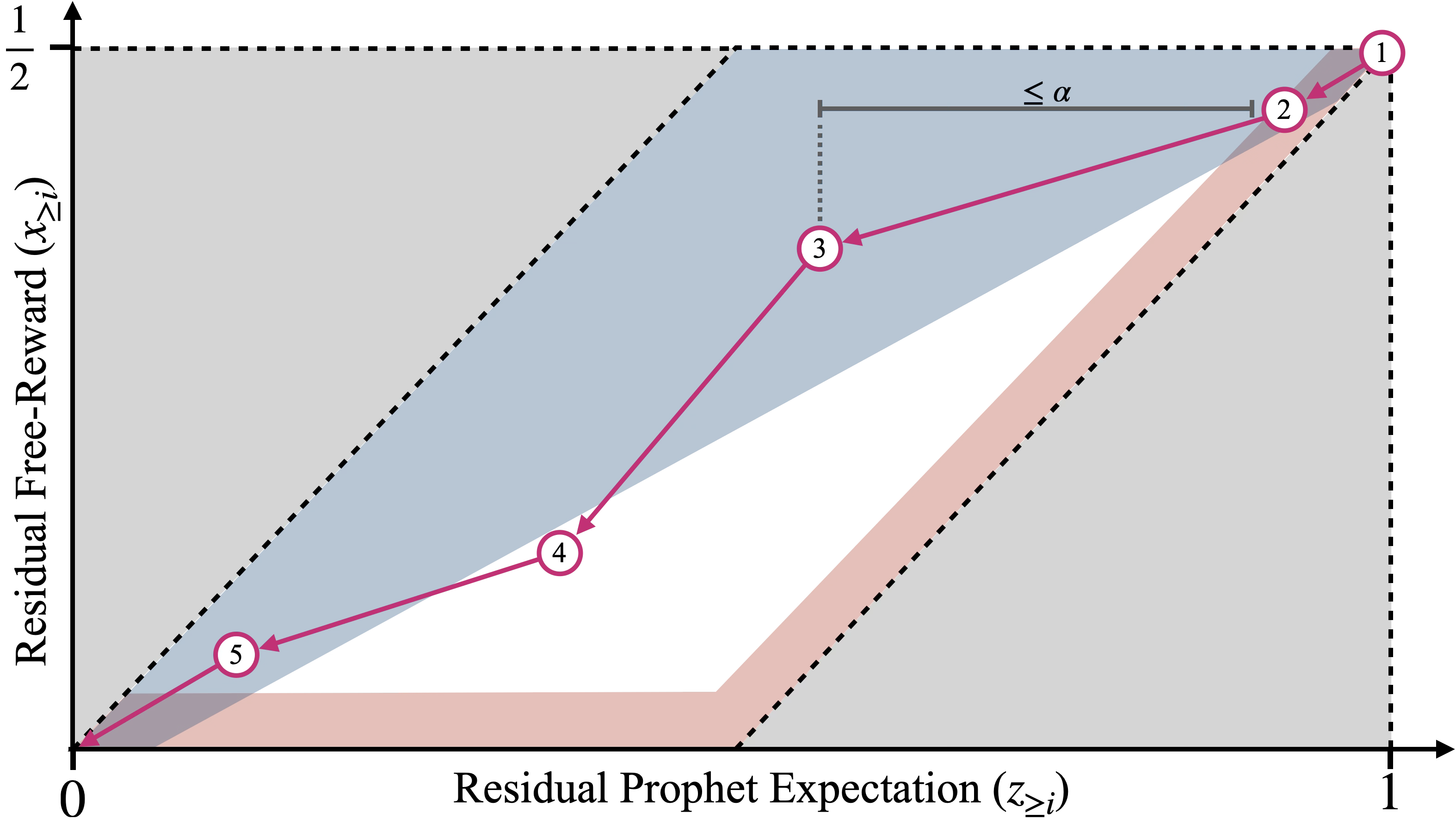}
    \caption[Beating $0.5$ when there is no Major-Box]{Beating $0.5$ when there is no Major-Box. Each step corresponds to a numbered point at $(z_{\ge i},x_{\ge i})$, where $z_{\ge i}$ (resp. $x_{\ge i}$) is the normalized residual prophet (resp. residual free-reward) from step $i$ onward (including). 
    The gray region is infeasible. In our
    analysis, we show the following properties: (1) If the reevaluation step occurs outside the blue area, then our algorithm has a better than $\frac{1}{2}$ guarantee of the prophet's expectation.
    (2) Outside the red area, the online optimum is bounded away from the prophet's expectation. 
    (3) Between each step, $z_{\ge i}$ decreases by at most $\alpha$ (otherwise, this box is a major-box). With these properties, we show that there must be a step whose point lies outside one of the blue and red areas in which our algorithm does the reevaluation step. This allows the algorithm to obtain an approximation ratio better than $0.5$ compared to the optimum online.}
    \label{fig:XZ-plane}
\end{figure}

Else, the instance is threshold-hard and it contains a major-box. Our algorithm uses the same high threshold (that captures the free-rewards), until one of two events occurs. Either the major-box has passed, or the residual prophet excluding the major-box is low enough. 
If the major-box has passed, the algorithm reevaluates the threshold as in the previously described case (to the better of half of the residual prophet or keeping the high threshold). 
If the residual prophet excluding the major-box is low, then the algorithm lowers its threshold to a carefully chosen threshold that allows obtaining a high fraction of the major-box's expectation (not just its free-reward component), while balancing the probability of accepting one of the other boxes. 

We use a parameterized algorithm to analyze the cases and their worst-case guarantees. We then optimize the parameters to obtain a better than $0.5$ approximation to all cases simultaneously against the optimum online. 
These cases from our analysis are presented in Figure~\ref{fig:decision tree}.

\begin{figure}[ht]
    \centering
    \includegraphics[width=0.85\textwidth]{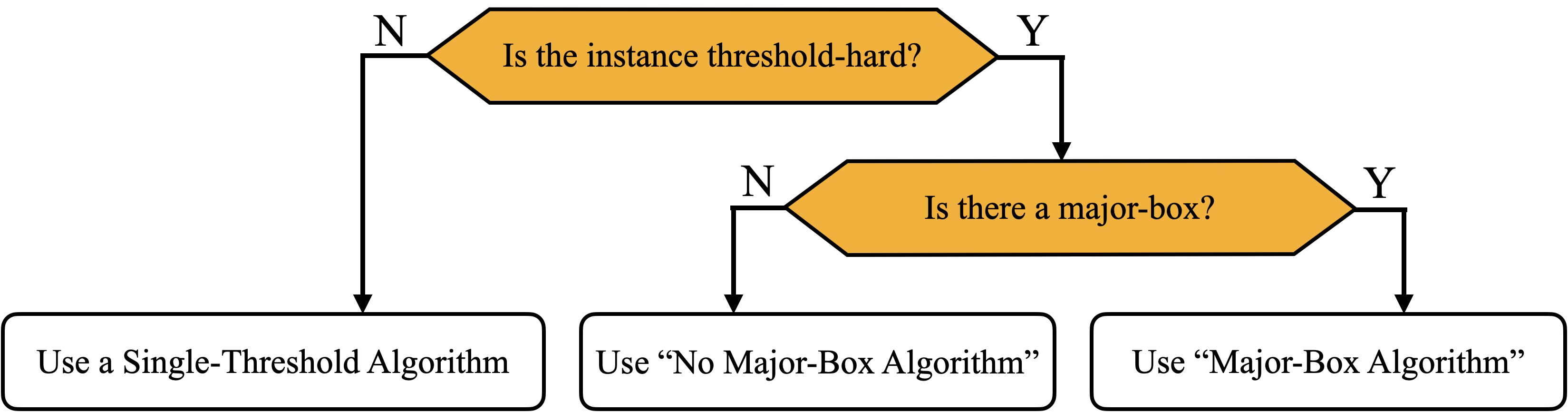}
    \caption[Decision Tree Diagram]{Decision Tree Diagram. ``No Major-Box Algorithm'' refers to Algorithm~\ref{alg:No Major Box Algorithm} and ``Major-Box Algorithm'' refers to Algorithm~\ref{alg:With Major Box Algorithm}}.
    \label{fig:decision tree}
\end{figure}

We complement this algorithmic result by showing an upper bound on the identity-truthfulness gap, that is strictly smaller than the currently known upper bound on the order-competitive ratio (which is $0.829$ \citep{chen:2024-settingtargets}). 
\begin{restatable}{theorem}{thirdtheorem}\label{thm:Hardness}
The identity-truthfulness gap is at most $0.81$, even for randomized algorithms.
\end{restatable}

Our last main result shows that among identity-truthful algorithms, there is a separation between (dynamic) pricing mechanisms and general algorithms (that can use non-monotone selection rules). Note that this is in sharp contrast to the case of non-identity-truthful algorithms for which there is no such separation~\citep{banihashem:2024} (even beyond single choice settings).

\begin{restatable}{theorem}{secondtheorem}\label{thm:Non-Monotone}
    There exists an instance for which the optimum identity-truthful algorithm does not adhere to a pricing mechanism.
\end{restatable}

This result provides the first separation in a prophet-like online selection setting between the power of general online algorithms and the power of pricing mechanisms. This addresses the question raised by \citet{dutting:2020} of whether prophet inequalities can always be implemented through pricing mechanisms, and whether there is a separation between anonymous (or more generally identity-truthful) vs. personalized (dynamic) pricing.

In Section~\ref{sec:single-threshold}, we use our technique for distinguishing between ``easy'' and ``threshold-hard'' instances to provide alternative proofs for classic prophet-inequality single-threshold algorithms.
In particular, we give alternative proofs that the median threshold, the half-mean threshold, the balanced-surplus threshold, the random prophet-sample threshold \citep{rubinstein:2020}, and a recent threshold by \citet{brustle:2026} are $\frac12$-competitive.
We also present several structural properties of thresholds that guarantee $\frac12$-competitiveness. In particular, we show that the set of quantiles $p \in [0,1]$ whose corresponding thresholds are $\frac12$-competitive contains an interval of length at least $\frac12$.
Our analysis answers an open question raised by \citet{waggoner2018prophet} concerning the characterization of all thresholds that guarantee $\frac12$-competitiveness.

\subsection{Related Work}
\paragraph{Information Gaps.}
Our work studies how information about reward identities affects the performance of online algorithms, contributing to a broader line of research on the role of information in online decision-making. This line of inquiry was initiated by \citet{azar2014prophet}, who introduced a framework for studying prophet inequalities when the algorithm has access only to samples rather than full distributional information. This sample-based model was later studied by \citep{rubinstein:2020, dutting:2021, caramanis:2022, gravin:2022, correa:2024, cristi:2024, fu:2024, ezra:2026-morethemerrier}.

More recently, \citet{ezra:2023-whoisnext} introduced the order-competitive ratio, which captures a different type of information gap: the effect of knowing the arrival order on the performance of online algorithms. This benchmark was subsequently adopted by \citep{chen:2024-settingtargets}, \citep{ezra:2023-importance}, and \citep{sun:2025}, who studied information gaps induced by arrival order in single-choice, combinatorial, and matching settings, respectively.
Later, \citet{ezra:2024-behindtheveil} introduced the identity-blindness gap which measures the importance of knowing the identities of the rewards, and showed that the advantage of comparing to the optimum online vanishes without this information.

\paragraph{Online Benchmarks.}
Another closely related line of research uses the optimum online as a benchmark for quantifying the approximation guarantees attainable under various algorithmic restrictions. \citet{niazadeh:2018} initiated this perspective by quantifying the performance loss incurred by restricting to single-threshold algorithms, relative to the optimal online algorithm which may use dynamic thresholds. They showed that single-threshold algorithms cannot guarantee a better than $1/2$ approximation even compared to the weaker benchmark of the optimum online.

Another type of constraint, initiated by \citet{papadimitriou:2021}, known as the philosopher inequality, compares the performance of polynomial-time algorithms to that of the optimum online. This problem has been studied in online matching \citep{papadimitriou:2021, saberi:2021, Braverman:2022, naor:2025, braverman:2025}, prophet secretary settings \citep{dutting:2023}, and combinatorial settings \citep{sun:2026}.

\paragraph{Fairness In Mechanism Design.}
Our work also relates to literature that incorporates fairness considerations in online decision-making.
\citet{correa:2021} and \citet{papasotiropoulos:2026} considered prophet and secretary variants under fairness constraints where the online decision-maker has to select items proportionally to prescribed group target ratios or group sizes.
\citet{arsenis:2022} defined two dimensions of fairness as maintaining the same probability of choosing a reward independent of identity or arrival time. \citet{balkanski:2024} defined a similar notion of fairness in a setting where for each box, a (possibly biased) prediction is given from the outset.
The main difference from our work is rather than restricting to a certain definition of fairness, we hide the discriminatory information (at the time of its relevant decision) thus forcing the algorithm to remain truthful.

\paragraph{Similar Techniques.}
A key step in our analysis is to study instances on which single-threshold algorithms perform poorly. In particular, we show that such instances approximately consist of deterministic boxes and ``free-reward'' boxes. A similar decomposition was used by \citet{sun:2025, sun:2026} for online decision-making problems with matching and combinatorial constraints.

Another important technique in our analysis is to isolate the largest box and design the algorithm according to both its contribution to the overall expected maximum reward and its unknown arrival position. A related idea appeared in \citet{chen:2025-matching} where they define the largest box as the item most likely to be selected by the prophet. They suggest an algorithm, designed around this box, where all other boxes are chosen according to activation rates while the largest box has a specialized activation probability.

\paragraph{Concurrent and Independent Work.}
Concurrent with the writing of this paper, \citet{correa:2026} independently developed their technique for analyzing the worst-case expected guarantees of a single-threshold algorithm.
Our approaches are similar, establishing the worst-case guarantee of a single-threshold algorithm based on the distribution of the prophet. Their analysis derives a function $V(t)$ that maps a threshold $t$ to the worst-case guarantee of the single-threshold algorithm with threshold $t$, whereas in our work we derive a similar function $\T(p)$ that maps the quantile $p\in[0,1]$ to the worst-case guarantee when using a threshold corresponding to the $p$-quantile of the prophet.
Denoting $F_{\max}$ as the CDF of the prophet, then for every threshold $t$ it holds that $\T(F_{\max}(t))=V(t)$, showing that the functions are identical (up to different coordinate system and normalization).

Both their paper and ours provide a unified framework for analyzing single-threshold algorithms that can derive alternative proofs for known prophet inequality results \textendash\ the guarantees of the median, half-mean and random prophet-sample thresholds.
Both papers prove that the range of $c$-competitive thresholds is an interval, however while $V(t)$ is quasi-concave, our function $\T(p)$ is concave due to the change of coordinate system. This allows us to derive several more properties, such as that the interval of quantiles that guarantee a competitive-ratio of $\frac12$ is at least of length $\frac12$.

\subsection{Paper Organization}
We first establish the model and notation in Section~\ref{sec:model}.
In Section~\ref{sec:gap}, we give a lower bound for the identity-truthfulness gap.
In Sections~\ref{subsec:single-thresh} and \ref{subsec:threshold-hard}, we show that ``easy'' instances have good guarantees using single-threshold algorithms and characterize the structure of worst-case instances.
In Section~\ref{subsec:optimum}, we use this characterization to obtain an upper bound for the optimum online benchmark.
In Section~\ref{subsec:algorithms}, we devise our parameterized algorithms and analyze their guarantees relative to the optimum online. By optimizing the parameters, we prove that the identity-truthfulness gap is strictly greater than $\frac12$.
In Section~\ref{sec:hardness}, we provide an upper bound on the identity-truthfulness gap, for any randomized algorithm.
In Section~\ref{sec:pricing}, we construct an instance that shows that the optimal identity-truthful algorithm cannot be implemented as a pricing mechanism.
In Section~\ref{sec:single-threshold}, we provide alternative proofs for known prophet inequality properties through the framework devised in this paper for the worst-case expectation of single-threshold algorithms, and prove additional properties of quantile-threshold algorithms.
Several technical claims and proofs are deferred to Appendix~\ref{apx:technical-claims}.

\section{Model}\label{sec:model}

\subsection{Preliminaries}
We consider a setting with $n$ boxes, each containing a stochastic reward. The rewards of the boxes are distributed according to known independent distributions $F_1,\ldots,F_n$. 
The boxes are ordered according to an arrival order $\pi$, where at step $i$, the reward of the $i$-th arriving box $v_i$ is drawn from distribution $F_{\pi^{-1}(i)}$, and revealed to a decision-maker. Upon observing the realized reward of box $i$, the decision-maker must immediately and irrevocably decide whether to accept the box. We will denote the value profile as $\vec{v}$ distributed according to $F=F_1\times\cdots\times F_n$.

In this paper, we consider two types of algorithms and compare their performance. One type is online algorithms that are order-aware and know $\pi$ from the outset. The optimum algorithm of this type calculates its thresholds via backward induction using $\pi$ and $F$. We use this algorithm's expectation as our benchmark.
The other type of algorithms is \textit{identity-truthful} algorithms that do not know $\pi$ from the outset, and their decision at $i$ cannot depend on the identity of box $i$ (it can only depend on the value $v_i$). 
However, in contrast to identity-blind algorithms (considered in \citep{ezra:2024-behindtheveil}), the algorithm learns the identity of box $i$ after deciding whether to accept its reward, and can use this information for future decisions.
Formally, the decision whether to accept reward $i$ can only depend on the realized values of $v_1,\ldots,v_i$, and on $\pi^{-1}(1),\ldots,\pi^{-1}(i-1)$.

Our goal is to explore the \textit{identity-truthfulness gap}, which is the worst-case ratio between the expected reward of an order-aware algorithm and an identity-truthful algorithm. 

We denote by $\OPT$ the optimum (order-aware) algorithm that maximizes the expectation of the chosen reward. For a product distribution $F$, an arrival order $\pi$, and a realization reward vector $\vec{v}$, we denote by  $\OPT(F,\pi;\vec{v})$ the reward chosen by $\OPT$. Notice that ``;'' is used to differentiate between inputs known from the outset (offline) and inputs revealed sequentially (online).

We denote an identity-truthful algorithm by $\ALG$. For a product distribution $F$, an arrival order $\pi$, and a realization reward vector $\vec{v}$, we denote by  $\ALG(F;\pi,\vec{v})$ the reward chosen by $\ALG$. 
\begin{definition}
    The \textbf{identity-truthfulness gap} of an identity-truthful algorithm \ALG\ is 
    $$\Gamma(\ALG)=\inf_F\min_\pi \frac{\E_{\vec{v}\sim F}[\ALG(F;\pi,\vec{v})]}{\E_{\vec{v}\sim F}[\OPT(F,\pi;\vec{v})]}.$$
\end{definition}
We omit $F$, $\pi$, and $\vec{v}$ from the notation when they are clear from context. 
When we say that the identity-truthfulness gap is at least $\rho$ without referring to a specific algorithm, we mean that there exists an identity-truthful algorithm, which we provide explicitly, whose identity-truthfulness gap is at least $\rho$. When we say that the identity-truthfulness gap is at most $\rho$, we mean that every identity-truthful algorithm has an identity-truthfulness gap at most $\rho$.

An important family of identity-truthful algorithms are single-threshold algorithms, denoted as $\STA{\tau}$, that use a fixed threshold $\tau\ge 0$ and accept the first reward exceeding it. The offline benchmark (a.k.a. prophet) is defined as $\E[\max_i v_i]$. It is a well established result that for all product distributions $F$ there exists $\tau$ such that $\forall_{\pi}:\E_{\vec{v}\sim F}[\STA{\tau}(F;\pi,\vec{v})]\ge \frac12\E[\max_i v_i]$ \citep{ester:1984}.

\subsection{Notations}
We establish the following notations used throughout the paper. Let us denote the prophet as $Z=\max_{i\in[n]} v_i$, and the maximum of the boxes starting at step $i$ as $Z_{\ge i}=\max_{i\le j\le n} v_j$.
We use $Q:[0,1]\to \mathbb{R}_{\geq 0}\cup\{\infty\}$ to denote the quantile function of the prophet $Z$, meaning $Q(p)=\sup \{x \mid \Pr[Z\leq x] = p\}$.
We assume without loss of generality that the distributions $\{F_i\}_{i\in[n]}$ have no point masses
\footnote{Any distribution with point masses can be approximated with a continuous distribution by replacing each point mass $x$ with a small uniform interval $[x,x+\epsilon]$ and taking the limit $\epsilon\to0$.}.
We further assume that the expectation of the prophet $\E[Z]$ equals $1$. This is possible since given a distribution, all values can be divided by the prophet's expectation in order to normalize the expectation to equal $1$. This normalization process does not impact the ratio between the expected rewards of our algorithms and the benchmark.

Our proofs use an approximation up to $\delta>0$ in order to bound various random value expectations and algorithms. Using a high threshold value drawn from the prophet's distribution $Q(1-\delta)$, each box $i$'s expectation is split into two parts in the following manner $$\E[v_i]=\E[v_i\cdot\ind{v_i\ge Q(1-\delta)}] + \E[v_i\cdot\ind{v_i< Q(1-\delta)}].$$ These two parts are denoted as $x_i=\E[v_i\cdot \ind{v_i\ge Q(1-\delta)}]$ and $y_i=\E[v_i\cdot \ind{v_i< Q(1-\delta)}]$. The total expectation of box $v_i$ is denoted as $z_i=x_i+y_i$. Similarly, $x_{\ge i}$, $y_{\ge i}$ and $z_{\ge i}$ are defined for $Z_{\ge i}$. Thus, for example, $z_{\ge 1}=1$. We term the expectation of the maximum of the boxes starting at step $i$, $z_{\ge i}$, the ``residual prophet'', and we term the free-reward component of this expectation, $x_{\ge i}$, the ``residual free-reward''.

These notations are useful when working with classic worst-case prophet inequality distributions constructed from two boxes \textendash\ a deterministic box, and a high-valued box with low realization probability (typically referred to as a free-reward box). The notation $x_i$ is associated with the free-reward component of box $i$.

\section{Identity-Truthfulness Gap Lower Bound}\label{sec:gap}

\subsection{Bounding the Performance of Single-Threshold Algorithm}\label{subsec:single-thresh}

First, we provide a lower bound for the expected performance of single-threshold algorithms on a given instance. 
While closely related to standard threshold/surplus decompositions, we phrase it
as a function of the prophet's quantile curve; this formulation is useful for our analysis.

\begin{claim}\label{cl:Single Threshold Algorithm Bound}
    The expected reward of a single-threshold algorithm with threshold $\tau$ is at least
    $$
    \E[\STA{\tau}] \ge p\int_p^1 Q(q)\frac{1}{q^2}dq,
    $$
    where $p:=\Pr[Z\le \tau]$.
\end{claim}

\begin{proof}
    The expectation of $\STA{\tau}$ can be decomposed into \textbf{revenue} and \textbf{surplus} in the following manner:
    \begin{equation}\label{eq:STA utility+revenue}
        \E[\STA{\tau}]=\tau\Pr[\STA{\tau} \text{ chooses}]+\E[(\STA{\tau}-\tau)^+]=\tau(1-p)+\int_{\tau}^{\infty}\Pr[\STA{\tau}>x]dx.
    \end{equation}
    Denoting $F_{\max}=\prod_{i\in[n]} F_i$, we can bound the second term as
    \begin{align}\label{eq:STA probabiliy bound}
        \forall_{x>\tau}:\Pr[\STA{\tau}>x]&=\sum_{i=1}^{n}(1-F_i(x))\left(\prod_{j<i}F_j(\tau)\right) \nonumber\\
        &= \sum_{i=1}^{n}(1-F_i(x))\left(\prod_{j<i}F_j(x)\frac{F_j(\tau)}{F_j(x)}\right) \nonumber\\
        &\ge \frac{F_{\max}(\tau)}{F_{\max}(x)}\sum_{i=1}^{n}(1-F_i(x))\left(\prod_{j<i}F_j(x)\right) = \frac{F_{\max}(\tau)}{F_{\max}(x)}(1-F_{\max}(x)).
    \end{align}
    Integrating Inequality~\eqref{eq:STA probabiliy bound}, we obtain
    \begin{align*}
        \int_{\tau}^{\infty}\Pr[\STA{\tau}>x]dx &\ge \int_{\tau}^{\infty} \frac{F_{\max}(\tau)}{F_{\max}(x)}(1-F_{\max}(x))dx \\
        &= F_{\max}(\tau)\int_{\tau}^{\infty} \left(\frac{1}{F_{\max}(x)}-1\right)dx\\
        &= p\int_p^1 \left(\frac{1}{q}-1\right)Q'(q)dq \\
        &= p\left[Q(q)(\frac{1}{q}-1)\right]_{p}^1+p\int_p^1 \frac{Q(q)}{q^2}dq \\
        &= p\int_p^1 \frac{Q(q)}{q^2}dq-\tau(1-p),
    \end{align*}
    where the second equality is given by the substitution $x=Q(q)$ and the third equality is given by applying integration by parts using $u=\frac{1}{q}-1$ and $v=Q(q)$.

    Finally, plugging this into Equation~\eqref{eq:STA utility+revenue} completes the proof.
\end{proof}
\begin{remark}
    We note that the inequality of Claim~\ref{cl:Single Threshold Algorithm Bound} can be an equality for every value of $\tau$. In particular, this holds in the instance with increasing weighted Bernoulli boxes with the same distribution of their maximum as $Z$. 
\end{remark}

We next define this lower bound on the performance of single-threshold algorithms as a function of the quantile $p$ that is set as a quantile threshold.\footnote{We note that the bound $\T(p)$ only accounts for the quantile function of the prophet, regardless of the individual boxes. Therefore, although an instance (defined by $F_1,\ldots,F_n$) may have a single-threshold algorithm with good guarantees, the guarantee given by $\T(p)$ can still be bad.}

\begin{definition}\label{def:t}
    Let $Z$ be a random variable corresponding to the prophet's distribution. We define a new function $\T: [0,1]\to \mathbb{R}_{\geq 0}$ as
    $$\T(p)=p\int_p^1 Q(q)\frac{1}{q^2}dq,$$
    i.e., the function receives a quantile $p$ and gives the lower bound for the expectation of a single-threshold algorithm with threshold $Q(p)$.
\end{definition}

The following claims show the properties of the function defined above. These properties are essential for Lemma~\ref{lem:CDF Lemma}.

\begin{claim}\label{cl:Concavity Claim}
    The function $\T(p)$ is concave. 
\end{claim}

\begin{proof}
    We will show that $\frac{\partial^2 \T}{\partial p^2} \le 0$.
    \begin{equation*}
        \frac{\partial^2 \T}{\partial p^2} =\frac{\partial^2}{\partial p^2}\left[p\int_p^1 \frac{Q(q)}{q^2}dq\right]=\frac{\partial}{\partial p}\left[\int_p^1 \frac{Q(q)}{q^2}dq-\frac{Q(p)}{p}\right] = -\frac{Q(p)}{p^2}-\frac{pQ'(p)-Q(p)}{p^2}=-\frac{Q'(p)}{p}\le0,
    \end{equation*}
    since the quantile function is non-decreasing and $p\ge 0$.
\end{proof}

\begin{claim}\label{cl:Area Claim}
    It holds that $\int_0^1 \T(p) dp = \frac12$. 
\end{claim}
\begin{proof}
    \begin{equation*}
        \int_0^1 \T(p)dp=\int_0^1 p\left[\int_p^1 \frac{Q(q)}{q^2}dq\right]dp=\int_0^1 \frac{Q(q)}{q^2}\left[\int_0^q pdp\right]dq
        =\int_0^1 \frac{Q(q)}{q^2}\left[\frac{q^2}{2}\right]dq=\frac{1}{2}\int_0^1 Q(q)dq=\frac12,
    \end{equation*}
    which concludes the proof.
\end{proof}

\begin{remark}
    We note that this is equivalent to the theorem that a single-threshold algorithm, that uses a sample of the prophet as its threshold, is $\frac12$-competitive \citep{rubinstein:2020}. The above proof is an alternative proof for this theorem, through the function $\T(p)$. 
\end{remark}

We next show that every concave function $\T:[0,1]\rightarrow \mathbb{R}_{\geq 0}$, with  $\T(1)=0$, can be obtained from some  quantile function $Q$. 
\begin{claim} \label{cl:allQ}
For every differentiable concave function $\T:[0,1]\rightarrow \mathbb{R}_{\geq 0}$ that satisfies $\T(1)=0$, there exists a quantile function $Q:[0,1]\rightharpoondown \mathbb{R}_{\geq 0}$ for which $\T$ is constructed by Definition~\ref{def:t} from $Q$.
\end{claim}
\begin{proof}
    We first notice that the derivative of $\T(p)$ constructed by Definition~\ref{def:t} is $\T'(p) = \int_p^1 \frac{Q(q)}{q^2} dq - \frac{Q(p)}{p} = \frac{1}{p}(\T(p) - Q(p))$.
    This implies that it holds that the quantile function $Q(p)=\T(p)-p\T'(p)$ corresponds to the function $\T(p)$ for the lower bound for the expectation of a single-threshold algorithm function. For this function $Q$ to be a valid quantile function, it must be non-negative and monotone. 
Indeed  $Q(0) = \T(0) \geq 0$, and  for monotonicity of $Q$, observe that  $$Q'(p)=\T'(p)-(\T'(p)+p\T''(p)) = p\T''(p) \geq 0 ,$$
where the inequality is by concavity of $\T$.
This concludes the proof.
\end{proof}

\subsection{\texorpdfstring{$\delta$}{δ}-Threshold-Hard Instances}\label{subsec:threshold-hard}
\label{sec:threshold-hard}
In this section, we use Definition~\ref{def:t} to detect easy instances for which single-threshold algorithms are sufficient. This allows us to design algorithms only for cases in which single-threshold algorithms might perform poorly (which we term $\delta$-threshold-hard).  
\begin{definition}
    A prophet's distribution (defined by its quantile function $Q$) is called \textbf{$\delta$-threshold-hard} if for every ${p\in[0,1]}$ it holds that $ \T(p) \le \frac12+\delta^2$.  
\end{definition}

We next provide various bounds on the quantile function $Q$ of the prophet for instances with $\delta$-threshold-hard prophet distributions.

\begin{lemma}\label{lem:CDF Lemma}
    Given $\delta\in(0,\frac14)$, if the prophet's distribution is \textit{$\delta$-threshold-hard}, then
    \begin{enumerate}[label=(\roman*)]
        \item $Q(\delta)\ge \frac12-2\delta$. 
        \item $x_{\ge 1}\ge \frac12-\delta$. (Recall $x_{\ge 1}=\E[Z\cdot\ind{Z\ge Q(1-\delta)}]$).
    \end{enumerate}

\end{lemma}

Informally, Lemma~\ref{lem:CDF Lemma} states that for instances that are $\delta$-threshold-hard, their prophet's expectation can be partitioned into two main parts \textendash\ values that are close to $Q(\delta)$ and values above the upper $\delta$-quantile (which their contribution is $x_{\geq 1}$). Both parts contribute close to half of the prophet's expectation and thus, those instances behave similarly to the known hard instance of prophet inequality (with a deterministic value of $1$, and $1/\epsilon$ with probability $\epsilon$).

\begin{proof}[Proof of Lemma~\ref{lem:CDF Lemma}]
    Recall the definition $\T(p) = p \int_p^1 \frac{Q(q)}{q^2} dq$.  We can express $Q(p)$ using the function $\T$ and it's derivative $\T'$ in the following manner
    \begin{align}\label{eq:Q as T and T'}
        \T'(p) &= \int_p^1 \frac{Q(q)}{q^2} dq - \frac{Q(p)}{p} = \frac{1}{p}(\T(p) - Q(p)) \nonumber\\
        &\implies Q(p)=\T(p)-p\T'(p).
    \end{align}
    We can also bound $\E[Z\cdot\ind{Z\ge Q(p)}]$ as
    \begin{align}\label{eq:FR T Bound}
        \T(p) &= p \int_{p}^1 \frac{Q(q)}{q^2} dq \le  \frac{1}{p} \int_{p}^1 Q(q) dq = \frac{1}{p}\E[Z\cdot\ind{Z\ge Q(p)}] \nonumber\\
        &\implies \E[Z\cdot\ind{Z\ge Q(p)}]\ge p\T(p).
    \end{align}
    Thus by bounding the values of $\T(p)$ and $\T'(p)$ at $\delta$ and $1-\delta$ we can achieve the necessary bounds for $Q(\delta)$ and $x_{\ge 1}$.

    Since we consider a $\delta$-threshold hard instance, we know that $\T(p)\leq \frac12+\delta^2$ for every $p$. Moreover, using concavity  (which is shown in Claim~\ref{cl:Concavity Claim}), we can bound $\T$ by its tangent at every point $p$. On the other hand, we know from Claim~\ref{cl:Area Claim} that the area under the function $\T$ must equal $\frac12$. The overall area under the upper bound on the maximum value of $\T$ (which is $\frac{1}{2}+\delta^2$ since this is a $\delta$-threshold hard instance) equals $\frac12+\delta^2$, which implies that the area between $\T$ and $y=\frac12+\delta^2$ is exactly $\delta^2$. This implies that the area above the tangent line (which is contained in it) cannot exceed $\delta^2$ (see figure~\ref{fig:threshold bound}). Let us denote the function of the tangent line at $\delta$ as $f(p)$. It holds that
    \begin{equation}\label{eq:T tangent bound}
        \T(p) \le f(p) = \T(\delta) + (p - \delta)\T'(\delta).
    \end{equation}

    \begin{figure}[ht]
        \centering
        \includegraphics[width=0.85\textwidth]{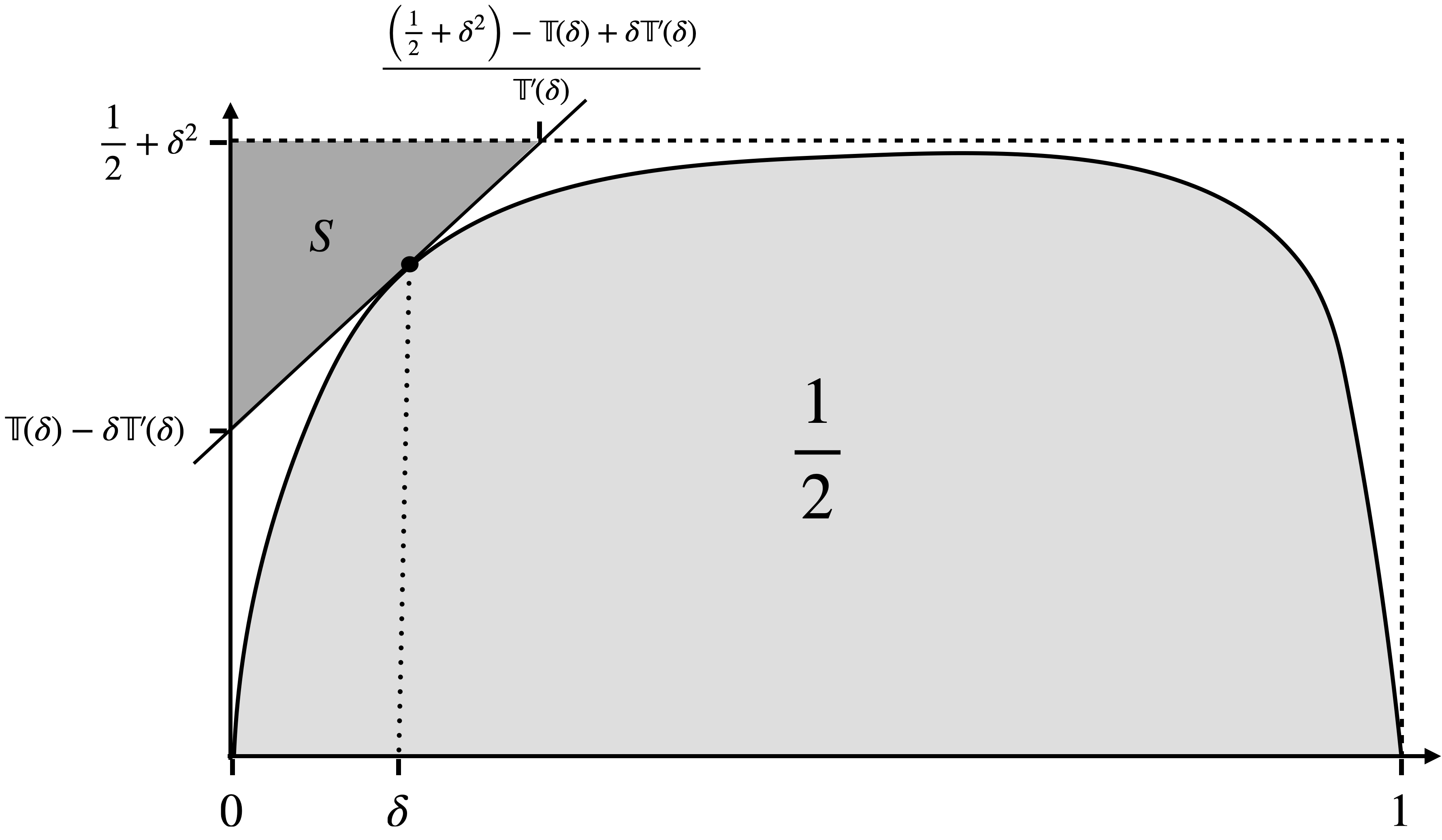}
        \caption[Illustration of $\T(p)$ and its tangent $f(p)$]{Illustration of $\T(p)$ and its tangent $f(p)$. The function $\T(p)$ is concave with area $\frac12$. For $\delta$-threshold-hard distributions, it is bounded by the rectangle $[0,1]\times\left[0,\frac12+\delta^2\right]$ as shown in the figure. The tangent line at $\delta$ intercepts with this rectangle, enclosing a triangle with area $S$ as depicted in the illustration.}
        \label{fig:threshold bound}
    \end{figure}

    This tangent intercepts $p=0$ at $f(0)=\T(\delta)-\delta \T'(\delta)\stackrel{\eqref{eq:Q as T and T'}}=Q(\delta)\ge 0$. It also intercepts the upper bound $\frac12+\delta^2$ at $p=\frac{\left(\frac{1}{2} + \delta^2\right) - \T(\delta) + \delta \T'(\delta)}{\T'(\delta)}$. When $\T'(\delta) \ge \frac{(\frac12+\delta^2)-\T(\delta)}{1-\delta}$ this interception point lies in the region $p\in[\delta,1]$, meaning the area bounded by the constant $\frac{1}{2} + \delta^2$ and the tangent line $f(p)$ forms a triangle as seen in Figure~\ref{fig:threshold bound}. We will later return to the case this does not hold.
    \paragraph{Case 1: $\T'(\delta) \ge \frac{(\frac12+\delta^2)-\T(\delta)}{1-\delta}$.} Denoting the area above the tangent at $\delta$ as $S$, its value is 
    \begin{align*}
        S =\frac12\cdot\left(\frac{1}{2} + \delta^2 - (\T(\delta) - \delta \T'(\delta)) \right)\cdot\frac{\frac{1}{2} + \delta^2 - \T(\delta) + \delta \T'(\delta)}{\T'(\delta)} \\
        =\frac{\left(\frac{1}{2} + \delta^2 - \T(\delta) + \delta \T'(\delta) \right)^2}{2\T'(\delta)}.
    \end{align*}
    As previously stated, it holds that $S \le \delta^2$. Rearranging for $\T(\delta)$, we obtain 
    \begin{equation}\label{eq:T bound by T'}
        \T(\delta) \ge \frac{1}{2} + \delta^2 - \sqrt{2\delta^2 \T'(\delta)} + \delta \T'(\delta).
    \end{equation}

    First, since $\T(\delta) \le \frac12+\delta^2$ we get
    \begin{equation}\label{eq:T' upper Bound}
        - \sqrt{2\delta^2 \T'(\delta)} + \delta \T'(\delta) \le 0 \implies \T'(\delta)\le 2.
    \end{equation}
    It holds that
    $$
    Q(\delta)\stackrel{\eqref{eq:Q as T and T'}}{=}\T(\delta)-\delta\T'(\delta) \stackrel{\eqref{eq:T bound by T'}}{\ge} \frac12+\delta^2-\sqrt{2\delta^2T'(\delta)} \stackrel{\eqref{eq:T' upper Bound}}{\ge} \frac{1}{2} - 2\delta + \delta^2\ge\frac12-2\delta,
    $$
    which proves the first inequality of the lemma for this case.

    Second, by taking the minimum of the RHS of Inequality~\eqref{eq:T bound by T'} over $\T'(\delta)$ (which is obtained at $\T'(\delta) = \frac{1}{2} $), we get that
    \begin{equation}\label{eq:T lower bound}
        \T(\delta) \ge \frac{1}{2} - \frac{\delta}{2} + \delta^2.
    \end{equation}
    Similarly, one can get the same lower bound for $\T(1-\delta)$ by a symmetric construction.
    Thus, we obtain the second inequality of the lemma for this case,
    $$
    x_{\ge 1}=\E[Z\cdot\ind{Z\ge Q(1-\delta)}] \stackrel{\eqref{eq:FR T Bound}}{\ge} (1-\delta) \T(1-\delta)\stackrel{\eqref{eq:T lower bound}}{\ge} (1-\delta)\left(\frac{1}{2} - \frac{\delta}{2} + \delta^2\right)\ge\frac12-\delta,
    $$
    where the last inequality holds since $\delta<\frac14$.
    
    \paragraph{Case 2: $\T'(\delta) <  \frac{\frac12+\delta^2-\T(\delta)}{1-\delta}$.} We now consider the case that the area formed above the tangent at $\delta$ does not form a triangle (within $[0,1]$). We can integrate the tangent bound in Inequality~\eqref{eq:T tangent bound} to get
    \begin{align*}
        \frac12 & = \int_0^1\T(p)dp \le \int_0^1\T(\delta)+(p-\delta)\T'(\delta)dp \\
        &= \T(\delta)+(\frac12-\delta)\T'(\delta) \\
        &<  \T(\delta)+(\frac12-\delta)\frac{\frac12+\delta^2-\T(\delta)}{1-\delta} \\
        &= \frac{\frac12}{1-\delta}\T(\delta)+\frac{\frac14-\frac12\delta+\frac12\delta^2-\delta^3}{1-\delta}.
    \end{align*}
    By rearrangement, we obtain that $\T(\delta)\ge \frac12-\delta^2+2\delta^3$ which implies (by the inequality of the case) that $\T'(\delta) < 2\delta^2 $. The last two inequalities that we obtained on $\T(\delta)$ and $\T'(\delta)$ imply (for $\delta<\frac{1}{4})$ Inequalities~\eqref{eq:T' upper Bound}~and~\eqref{eq:T lower bound}. The rest of the proof follows the same steps as in Case 1, by using these inequalities. This concludes the proof of the lemma.
\end{proof}

We next provide additional upper bounds on the shape of $\delta$-threshold-hard prophet distributions. This includes showing that the concentration of values between the lower and upper $\delta$-quantiles is of order $O(\delta)$.

\begin{claim}\label{cl:Q Upper Bound}
    Given $\delta\in(0,\frac14)$, if the prophet's distribution is \textit{$\delta$-threshold-hard}, then 
    $$
    Q(\delta)\le \frac12+2\delta^2.
    $$
\end{claim}
\begin{proof}
    Since $\T'(p) = \int_p^1 \frac{Q(q)}{q^2} dq - \frac{Q(p)}{p} = \frac{1}{p}(\T(p) - Q(p))$, we have $Q(\delta) = \T(\delta) - \delta \T'(\delta)$. Similarly to Lemma~\ref{lem:CDF Lemma}, we can integrate over the tangent bound as such
    \begin{eqnarray}
            \frac12 & =&\int_0^1\T(p)dp\stackrel{\eqref{eq:T tangent bound}}\le\int_0^1\T(\delta)+(p-\delta)\T'(\delta)dp \nonumber \\
            &=&\T(\delta)+(\frac12-\delta)\T'(\delta) \nonumber\\
            &\le &\frac12+\delta^2+\left(\frac12-\delta\right)\T'(\delta) \nonumber\\
            &\implies &\T'(\delta)\ge-\frac{\delta^2}{\frac12-\delta},
    \end{eqnarray}
    where the first equality is by Claim~\ref{cl:Area Claim} and the second inequality follows since the prophet's distribution is $\delta$-threshold-hard. Therefore, $$Q(\delta)=\T(\delta) - \delta \T'(\delta)\le\frac12+\delta^2+\delta\frac{\delta^2}{\frac12-\delta}\le\frac12+2\delta^2,$$
    where the last inequality holds since $\delta<\frac14$, concluding the proof.
\end{proof}

\begin{claim}\label{cl:FR Upper Bound}
    Given $\delta\in(0,\frac14)$, if the prophet's distribution is \textit{$\delta$-threshold-hard}, then 
    $$
    x_{\ge 1}\le \frac12+\delta.
    $$
\end{claim}
\begin{proof}
    We have, $$\frac12+\delta^2 \ge \T(1-\delta)= (1-\delta)\int_{1-\delta}^1 \frac{Q(q)}{q^2}dq \ge (1-\delta)\int_{1-\delta}^1 Q(q)dq = (1-\delta)x_{\ge 1}.$$ By rearrangement we get that $$x_{\ge 1} \le \frac{\frac12+\delta^2}{1-\delta}\le \frac{1}{2} + \delta,$$
    where the last inequality holds since $\delta<\frac14$, concluding the proof.
\end{proof}

\begin{claim}\label{cl:Mid-Range Upper Bound}
    Given $\delta\in(0,\frac14)$, if the prophet's distribution is \textit{$\delta$-threshold-hard}, then 
    $$
    \E\left[(Z-Q(\delta))\cdot\ind{Q(\delta)\le Z<Q(1-\delta)}\right] \le 4\delta.
    $$
\end{claim}
\begin{proof}
    It holds that
    \begin{align*}
        \E\left[(Z-Q(\delta))\cdot\ind{Q(\delta)\le Z<Q(1-\delta)}\right]
        &= \E\left[Z\cdot\ind{Q(\delta)\le Z<Q(1-\delta)}\right] - \E\left[Q(\delta)\cdot\ind{Q(\delta)\le Z<Q(1-\delta)}\right] \\
        &= \E\left[Z\cdot\ind{Z\ge Q(\delta)}\right] - \E\left[Z\cdot\ind{Z\ge Q(1-\delta)}\right] - Q(\delta)\Pr\left[Q(\delta)\le Z<Q(1-\delta)\right] \\
        &\le \E[Z] -x_{\ge 1} - (1-2\delta)Q(\delta) \\
        &\le 1 - \left(\frac12-\delta\right)-(1-2\delta)\left(\frac12-2\delta\right)=4\delta-4\delta^2\le4\delta,
    \end{align*}
    where the first inequality holds since $\E\left[Z\cdot\ind{Z\ge Q(\delta)}\right]\le\E[Z]$ and the second inequality holds by Lemma~\ref{lem:CDF Lemma}.
\end{proof}

To summarize, in this section, we bound three key features of $\delta$-threshold-hard instances. These features correspond to the ``deterministic'' component, the ``free-reward'' component, and the expectation of the values in between. Overall, we can summarize it as:
\begin{eqnarray*}
    -2\delta \le &Q(\delta)-\frac12& \le 2\delta^2 \quad   \left(\mbox{the deterministic component is close to $\frac{1}{2}$}\right) \\
    -\delta \le &x_{\ge 1}-\frac12& \le \delta \quad   ~~~\left(\mbox{the free-reward component is close to $\frac{1}{2}$}\right)\\
    0 \le &\E\left[(Z-Q(\delta))\cdot\ind{Q(\delta)\le Z<Q(1-\delta)}\right]& \le 4\delta \quad  ~~ \left(\mbox{non-free-rewards beyond $Q(\delta)$ contribute at most $O(\delta)$}\right)
\end{eqnarray*}

\subsection{Bounding the Performance of the Optimum  Online}\label{subsec:optimum}
We next utilize the claims from Section~\ref{sec:threshold-hard} to bound the performance of the optimum online algorithm.

\begin{lemma}\label{lem:Optimal Online Upper Bound}
    Given $\delta\in(0,\frac14)$, if the prophet's  distribution is \textit{$\delta$-threshold-hard}, then  for every $i\in[n]$ it holds that
    $$\E[\OPT]\le \frac{x_{\ge 1}-x_{\ge i}}{1-\delta}+\max\{Q(\delta),z_{\ge i}\}+\frac{4\delta}{1-\delta}.$$
\end{lemma}
Intuitively, this claim states that for $\delta$-threshold-hard cases, up to low order terms (of order of $\delta$), the performance of the optimum online algorithm can be bounded as the sum of two terms, the ``free-reward'' of the first $i-1$ boxes (captured by $x_{\geq 1}-x_{\geq i}$), and the maximum between a deterministic box with a value $Q(\delta)$ and the residual prophet (captured by $\max\{Q(\delta),z_{\ge i}\}$). This bound holds for every $i$.
\begin{proof}[Proof of Lemma~\ref{lem:Optimal Online Upper Bound}]
    Let $i\in[n]$, consider an auxiliary instance with the following two boxes: $$B_1 = \max\{\max_{1\le j<i}v_j,\ Q(\delta)\},\quad B_2=\max_{i\le j\le n} v_j,$$ meaning $B_1$ is the maximum of the first $i-1$ boxes and a value of $Q(\delta)$, and $B_2$ is the maximum of the remaining boxes, starting from $i$. The optimum online algorithm for this auxiliary instance receives in expectation at least the expectation of the optimum online algorithm of the original instance\footnote{The expected reward of the optimum online can only increase when merging boxes and replacing them by their maximum. The same holds for adding a deterministic value of $Q(\delta)$ to the maximum.}. The optimum online algorithm for the auxiliary instance (which we denote by $\overline{\OPT}$) is to accept the reward from $B_1$ if and only if it exceeds the threshold $\E[B_2]$. This guarantees
    \begin{equation}\label{eq:optimum value}
        \E\left[\overline{\OPT}\right] = \E\left[B_2\right]+\E\left[\left(B_1-\E\left[B_2\right]\right)^+\right].
    \end{equation}

    Let us first consider the following decomposition
    \begin{align}\label{eq:B1 decomposition}
        \E\left[(B_1-Q(\delta))^+\right] &\le \E\left[(B_1-Q(\delta))\cdot \ind{Q(\delta)\le B_1<Q(1-\delta)}\right]+\E\left[(B_1-Q(\delta))\cdot \ind{B_1\ge Q(1-\delta)}\right] \nonumber\\
        &\le \E\left[(B_1-Q(\delta))\cdot \ind{Q(\delta)\le B_1<Q(1-\delta)}\right]+\E\left[B_1\cdot \ind{B_1\ge Q(1-\delta)}\right].
    \end{align}
    We can bound the first term of Inequality~\eqref{eq:B1 decomposition} by 
    \begin{align*}
        \E\left[(B_1-Q(\delta))\cdot \ind{Q(\delta)\le B_1<Q(1-\delta)}\right] 
        &\le \frac{\E\left[(\max\{B_1,B_2\}-Q(\delta))\cdot\ind{Q(\delta)\le\max\{B_1,B_2\}< Q(1-\delta)}\right]}{\Pr[\max\{B_1,B_2\}<Q(1-\delta)]} \nonumber\\
        &= \frac{\E\left[(Z-Q(\delta))\cdot\ind{Q(\delta)\le Z< Q(1-\delta)}\right]}{\Pr[Z<Q(1-\delta)]}\nonumber\\
        &\le \frac{4\delta}{1-\delta},
    \end{align*}
    where the first inequality holds by Claim~\ref{cl:General Mid-Range Integral} (assigning $X=B_1$, $Y=B_2$, $a=Q(\delta)$, $b=Q(1-\delta)$), the first equality holds since $Z=\max\{B_1,B_2\}$ for values above $Q(\delta)$, and the second inequality holds by Claim~\ref{cl:Mid-Range Upper Bound}.

    For the second term of Inequality~\eqref{eq:B1 decomposition}, it holds that
    \begin{align*}
        \E\left[B_1\cdot \ind{B_1\ge Q(1-\delta)}\right]
        &\le \frac{\E\left[\max\{B_1,B_2\}\cdot \ind{\max\{B_1,B_2\}\ge Q(1-\delta)}\right]-\E\left[B_2\cdot \ind{B_2\ge Q(1-\delta)}\right]}{\Pr[\max\{B_1,B_2\}<Q(1-\delta)]} \nonumber\\
        &= \frac{\E\left[Z\cdot \ind{Z\ge Q(1-\delta)}\right]-\E\left[B_2\cdot \ind{B_2\ge Q(1-\delta)}\right]}{\Pr[Z<Q(1-\delta)]} \nonumber\\
        &= \frac{x_{\ge 1}-x_{\ge i}}{1-\delta},
    \end{align*}
    where the first inequality holds by Claim~\ref{cl:General FR Additiveness} (assigning $X=B_1$, $Y=B_2$, $T=Q(1-\delta)$), the first equality holds since $Z=\max\{B_1,B_2\}$ for values above $Q(\delta)$, and the second equality holds by the definition of $B_2$, $x_{\ge 1}$ and $x_{\ge i}$. 
    Using the bounds of these two terms of Inequality~\eqref{eq:B1 decomposition} implies that 
    \begin{equation}\label{eq:B1 bound}
        \E[(B_1-Q(\delta))^+]\le \frac{x_{\ge 1}-x_{\ge i}}{1-\delta} + \frac{4\delta}{1-\delta}.
    \end{equation}
    
    We now consider two cases (whether $\E[B_2]\le Q(\delta)$ or not).
    \paragraph{Case 1: $\E[B_2]\le Q(\delta)$.}
    In this case, $B_1$ is always chosen, therefore $\E[\overline{\OPT}]=\E[B_1]$. Since $B_1$ is at least $Q(\delta)$, it holds that 
    \begin{equation}\label{eq:B1 expectation}
        \E[B_1]=\E[(B_1-Q(\delta))^+]+Q(\delta).
    \end{equation}
    Overall we have that  $$\E[\OPT]\le \E[\overline{\OPT}] = \E[B_1]\le \frac{x_{\ge 1}-x_{\ge i}}{1-\delta}+Q(\delta)+\frac{4\delta}{1-\delta},$$
    where the first inequality holds since $\overline\OPT$ obtains at least as $\OPT$, the equality holds since we are in the case where $Q(\delta)\geq \E[B_2]$, and the second inequality holds by Inequalities~\eqref{eq:B1 bound} and \eqref{eq:B1 expectation}.

    \paragraph{Case 2: $\E[B_2] > Q(\delta)$.}
    In this case, we have $$\E[\OPT] \le \E\left[\overline{\OPT}\right]= E[B_2]+\E\left[(B_1-\E[B_2])^+\right]\le \E[B_2]+\E\left[(B_1-Q(\delta))^+\right]\le \frac{x_{\ge 1}-x_{\ge i}}{1-\delta}+z_{\ge i}+\frac{4\delta}{1-\delta},$$
    where the first inequality holds since $\overline\OPT$ obtains at least as $\OPT$; the first equality holds from Equation~\eqref{eq:optimum value}; the second inequality holds since if $\E[B_2] > Q(\delta)$ then $\E\left[(B_1-\E[B_2])^+\right]\le \E\left[(B_1-Q(\delta))^+\right]$; the last inequality holds by Inequality~\eqref{eq:B1 bound}.
\end{proof}

\subsection{Algorithms for \texorpdfstring{$\delta$}{δ}-Threshold-Hard Instances}\label{subsec:algorithms}

In this section, we provide two identity-truthful algorithms for instances with $\delta$-threshold-hard prophet distributions. 
Our first algorithm (Algorithm~\ref{alg:No Major Box Algorithm}) handles instances in which no box is too large, i.e., all boxes have an expectation of at most $\alpha$ for some parameter $\alpha$. 
Our second algorithm (Algorithm~\ref{alg:With Major Box Algorithm}) handles instances where there is such a box, which we term a ``major-box.''

\paragraph{No Major-Boxes.}
We first handle instances where there is no major-box. Formally, we assume that for every $i$, it holds that
$$ \E[v_i] \leq \alpha.$$
Our algorithm starts with a high-quantile threshold of $Q(1-\delta)$ to capture high values. Once the residual prophet is low enough (below $\frac{1+\alpha}{2}$), the algorithm chooses the better between guaranteeing the residual free-rewards (by keeping the high threshold) or switching to a different threshold obtaining half of the residual prophet.
In instances with no major-box, the residual prophet decreases slowly enough with each step such that there must be a step in which the better of the two guarantees yields a good approximation of the optimum online.

\begin{lemma}\label{lem:No Major Box Algorithm Bound}
    Given parameters $\delta\in(0,\frac14)$ and $\alpha\in(\frac34,1)$, if the prophet's distribution is \textit{$\delta$-threshold-hard} and $\max_{i\in[n]} \E[v_i] \le \alpha$, then Algorithm~\ref{alg:No Major Box Algorithm} achieves $$\frac{\E[\ALG]}{\E[\OPT]} \ge \frac{2}{\alpha+3}-16\delta.$$
    
    \begin{algorithm}
        \caption{No Major-Box Algorithm}\label{alg:No Major Box Algorithm}
        \begin{algorithmic}[1]
            \State Set the current threshold to $\tau \gets Q(1-\delta)$.
            \While{$z_{\ge i} > \frac{1+\alpha}{2}$} \label{code:remaining check}
                \State Observe the next box $i$.
                \If{$v_i \ge \tau$}
                    \State Accept reward $i$ and stop.
                \EndIf
            \EndWhile 
            \If{$\frac12 z_{\ge i} \ge x_{\ge i}$} \label{code:choose threshold}
                \State Set the current threshold to $\tau \gets \frac12 z_{\ge i}$. \label{code:take half}
            \Else
                \State Keep the current threshold $\tau \gets Q(1-\delta)$. \label{code:keep threshold}
            \EndIf
            \State Continue observing boxes, accepting the first reward $i$ where $v_i \ge \tau$.
        \end{algorithmic}
    \end{algorithm}
\end{lemma}

\begin{proof}
    Let us denote the index of the switching step in Algorithm~\ref{alg:No Major Box Algorithm}, where the condition in line~\ref{code:remaining check} isn't satisfied, as $i$. 
    Such a step $i$ is guaranteed to exist since for the last box $n$, we have $z_{\ge n}=z_{n}\le\alpha<\frac{1+\alpha}{2}$.

    We first give a lower bound for the expectation of $\ALG$. Notice in the branching at Line~\ref{code:choose threshold}, if Line~\ref{code:keep threshold} is reached then the algorithm is a single-threshold algorithm guaranteeing
    \begin{align}\label{eq:single threshold guarantee}
        \E[\ALG]=\E[\STA{\tau}] &\ge \tau\Pr[Z\ge \tau]+\Pr[Z<\tau]\E[(Z-\tau)^+] \nonumber\\
        &=\tau\Pr[Z\ge \tau]+\Pr[Z<\tau]\left(\E[Z\cdot\ind{Z\ge\tau}]-\tau\Pr[Z\ge\tau]\right) \nonumber\\
        &=\tau\Pr[Z\ge \tau]^2+\Pr[Z<\tau]\E[Z\cdot\ind{Z\ge\tau}] \nonumber\\
        &\ge\Pr[Z<\tau]\E[Z\cdot\ind{Z\ge\tau}] \nonumber\\
        &=(1-\delta)x_{\ge 1},
    \end{align}
    where the first inequality is a known revenue + surplus decomposition, and the last equality holds because $\tau=Q(1-\delta)$.

    If Line~\ref{code:take half} is reached, then up until step $i$ (not included), by Claim~\ref{cl:General FR Additiveness} (assigning $X=\max_{1\leq j\leq i-1}v_j$, $Y=\max_{i\le j\le n}v_j$ and $T=Q(1-\delta)$), the algorithm achieves at least $(1-\delta)(x_{\ge 1}-x_{\ge i})$. The algorithm stops before step $i$ with probability of at most $\delta$. From step $i$, by taking half of the residual prophet's expectation as a threshold, the algorithm obtains at least $\frac12 z_{\geq i}$ \citep{kleinberg2012matroid}. Overall, in this branch, the algorithm guarantees $\E[\ALG]\ge(1-\delta)(x_{\ge 1}-x_{\ge i})+(1-\delta)\frac12z_{\ge i}.$
    Since the condition of Line~\ref{code:choose threshold} is choosing the maximum of these two guarantees, the algorithm overall obtains at least: $$\E[\ALG]\ge\max\{(1-\delta)(x_{\ge 1}-x_{\ge i})+(1-\delta)\frac12z_{\ge i},(1-\delta)x_{\ge 1}\}=(1-\delta)\left(x_{\ge 1} + \left[\frac{z_{\ge i}}{2}-x_{\ge i}\right]^+\right).$$
    By the upper bound on $\E[\OPT]$ in Lemma~\ref{lem:Optimal Online Upper Bound}, it holds that
    \begin{align}\label{eq:ratio-lem2}
        \frac{\E[\ALG]}{\E[\OPT]}
    &\ge \frac{x_{\ge 1} + \left[\frac{z_{\ge i}}{2}-x_{\ge i}\right]^+}{\frac{x_{\ge 1}-x_{\ge i}}{1-\delta}+\max\{Q(\delta), z_{\ge i}\}+\frac{4\delta}{1-\delta}}(1-\delta) \nonumber\\
    &\ge \frac{\frac12+\left[\frac{z_{\ge i}}{2}-x_{\ge i}\right]^+-\delta}{\frac12-x_{\ge i}+\max\{\frac12,z_{\ge i}\}+6\delta}(1-\delta)^2 \nonumber\\
    &\ge \frac{\frac12+\left[\frac{z_{\ge i}}{2}-x_{\ge i}\right]^+}{\frac12-x_{\ge i}+\max\{\frac12,z_{\ge i}\}}-16\delta,
    \end{align}
    where the second inequality holds by Lemma~\ref{lem:CDF Lemma}, Claims \ref{cl:FR Upper Bound} and \ref{cl:Q Upper Bound} and since $\delta<\frac14$, and the third inequality holds for $0\le x_{\ge i}\le z_{\ge i}\le1$.

    We next show that $z_{\ge i}\in[\frac{1-\alpha}{2},\frac{1+\alpha}{2}]$.
    Consider the sequence $z_{\ge 1},\ldots,z_{\ge n}$. This sequence is non-increasing with $z_{\ge 1}=1$, where each difference between $z_{\geq j-1}$ and $z_{\geq j}$ is bounded by $z_j$ which is at most $\alpha $ (since there are no major-boxes).
    By definition of $i$, it holds that $z_{\ge i-1}>\frac{1+\alpha}{2}\ge z_{\ge i}$, which implies that $z_{\ge i}\in[\frac{1-\alpha}{2},\frac{1+\alpha}{2}]$. 

    By Claim~\ref{cl:Function Optimization}, and since $z_{\ge i}\in[\frac{1-\alpha}{2},\frac{1+\alpha}{2}]$, we have that the minimum of the RHS of Equation~\eqref{eq:ratio-lem2} must be either at $(x_{\ge i},z_{\ge i})=( \frac{1-\alpha}{4},\frac{1-\alpha}{2})$ or at $(x_{\ge i},z_{\ge i})=(\frac{1+\alpha}{4},\frac{1+\alpha}{2})$.
       The value at these two points is $$\frac{2}{\alpha+3}-16\delta,$$ concluding the proof.
\end{proof}

\paragraph{There is a Major-Box.} We now handle instances where there is a major box.  
We note that our chosen parameters ($\alpha$ and $\delta$) ensure that there can only be one major-box (See  Claim~\ref{cl:one major box} in the appendix).
As in the ``no major-box case'', our algorithm first takes only high values above $Q(1-\delta)$ as long as the major-box hasn't passed (by step $i-1$), and the expectation of the remaining boxes (excluding the major-box) is high enough (above a parameter $\beta$). 
If the major-box arrives, the algorithm changes to the better of the two thresholds of $\frac12 z_{\geq i}$ and $Q(1-\delta)$ (as in Algorithm~\ref{alg:No Major Box Algorithm}). If the expectation of the remaining boxes (excluding the major-box) reaches a low enough value, then the algorithm switches to a less risky threshold of $\sqrt{\alpha \beta}$.
This threshold obtains more of the value of the major-box (including values below $Q(1-\delta)$) without increasing the probability of selecting one of the non major-boxes too much (which might lead to losing the potential high values of the major-box).

\begin{lemma}\label{lem:With Major Box Algorithm Bound}
    Given parameters $\delta\in(0,0.02)$, $\alpha\in(\frac34,1)$, and $\beta\in(0,\frac\alpha4)$, if the prophet's distribution is \textit{$\delta$-threshold-hard} and $\max_i z_i \ge \alpha$, then Algorithm~\ref{alg:With Major Box Algorithm} achieves $$\frac{\E[\ALG]}{\E[\OPT]}\ge\min\left\{\frac{1}{2-\beta}-16\delta, \alpha+2\beta-2\sqrt{\alpha\beta}\right\}.$$
    \begin{algorithm}
        \caption{Major-Box Algorithm}\label{alg:With Major Box Algorithm}
        \begin{algorithmic}[1]
            \State Set $m \gets \arg\max_j \E_{v\sim F_j}[v]$. \label{code:major-box notation}
            \State Set the current threshold to $\tau \gets Q(1-\delta)$.
            \While{$\pi^{-1}(i-1) \neq m$ \textbf{and} $\E[\max_{j\ge i\wedge j\neq \pi(m)} v_j] > \beta$} \label{code:re-evaluate condition}
                \State Observe the next box $i$.
                \If{$v_i \ge \tau$}
                    \State Accept reward $i$ and stop.
                \EndIf
            \EndWhile 
            
            \If {$\pi^{-1}(i-1)=m$} \label{code:major-box passed}
                \If{$\frac12 z_{\ge i}\ge x_{\ge i}$}
                    \State Set the current threshold to $\tau \gets \frac12 z_{\ge i}$.
                \Else
                    \State Keep the current threshold $\tau \gets Q(1-\delta)$.
                \EndIf
            \Else
                \State Set the current threshold to $\tau \gets \sqrt{\alpha\beta}$.
            \EndIf
            \State Continue observing boxes, accepting the first reward $i$ where $v_i \ge \tau$.
        \end{algorithmic}
    \end{algorithm}
\end{lemma}

\begin{proof}
    As in Line~\ref{code:major-box notation} of Algorithm~\ref{alg:With Major Box Algorithm}, let us denote $m:= \arg\max_j \E_{v\sim F_j}[v]$ as the index of the distribution of the major-box $v_{\pi(m)}$. This means that the major-box arrives at step $\pi(m)$.
    The algorithm continues as long as $v_{\pi(m)}$ has not passed and $\E[\max_{j\ge i \wedge j \neq\pi(m)} v_j] > \beta$ (Line~\ref{code:re-evaluate condition}).
     One of these conditions must fail before all boxes have passed since if $\pi(m)<n$ the first condition fails at step $\pi(m)+1$, and if $\pi(m)=n$ then the second condition must fail by the last step (included), since $\E[\max_{j\geq n \wedge j \neq\pi(m)} v_j]= 0 < \beta$.
     We further note that these conditions are mutually exclusive. If box $\pi(m)$ has passed at step $i-1$ the remaining expectation stays the same as in the previous step (which was more than $\beta$), and if the expectation of the remaining boxes (excluding $\pi(m)$) decreases below $\beta$, then a box other than $\pi(m)$ has passed.
     We can bound the performance of the algorithm by considering two cases (depending on which of the conditions fails first):
     
    \paragraph{Case 1: $\pi^{-1}(i-1)=m$.}
    As stated, if the algorithm reaches step $i$ where $\pi^{-1}(i-1)=m$ then $z_{\ge i}=\E[\max_{j\ge i} v_j]=\E[\max_{j\ge i-1\wedge j \neq\pi(m)} v_j]>\beta$. In this case, Algorithm~\ref{alg:With Major Box Algorithm} behaves similarly to Algorithm~\ref{alg:No Major Box Algorithm}. Consequently, by the same derivations leading to Inequality~\eqref{eq:ratio-lem2} (from the proof of Lemma~\ref{lem:No Major Box Algorithm Bound}), we have 
    \begin{equation}\label{eq:ratio-lem4}
        \frac{\E[\ALG]}{\E[\OPT]}
        \ge \frac{\frac12+\left[\frac{z_{\ge i}}{2}-x_{\ge i}\right]^+}{\frac12-x_{\ge i}+\max\{\frac12,z_{\ge i}\}}-16\delta.
    \end{equation}
    By Claim~\ref{cl:Function Optimization}, the RHS of Inequality~\eqref{eq:ratio-lem4} is minimized at $x_{\geq i} = \frac12z_{\geq i}$ and lies on the extreme values of $z_{\ge i}$. 
    We next bound $z_{\geq i}$. 
    
    For the lower bound, in this case it holds that $z_{\ge i}>\beta$. Therefore $(x_{\ge i},z_{\ge i})=(\frac12\beta,\beta)$ is a possible minimum point, for which the value of the RHS of Equation~\ref{eq:ratio-lem4} is
    \begin{equation}\label{eq:beta minimum}
        \frac{1}{2-\beta}-16\delta.
    \end{equation}

    We next bound $z_{\ge i}$ from above to show that the other possible minima is always greater, for the given parameters. It holds that
    \begin{equation}\label{eq:z bound}
        z_{\ge i}=y_{\ge i}+x_{\ge i}\le \frac{y_{\ge 1}}{1-\delta}+x_{\ge i} \le \frac{\frac12+\delta}{1-\delta}+x_{\ge i}\le\frac12+x_{\ge i}+2\delta,
    \end{equation}
    where the first inequality is by Claim~\ref{cl:Semi Deterministic Monotonicity} (assigning $X=\max_{i\le j\le n}v_j$, $Y=\max_{1\le<j\le i-1}v_j$ and $T=Q(1-\delta)$), the second inequality is by Lemma~\ref{lem:CDF Lemma}, and the third inequality holds for $\delta<\frac14$.
    
    To bound $x_{\ge i}$ we have that
    \begin{multline}\label{eq:x bound}
        x_{\ge i}\le x_{\ge i-1}-(1-\delta)x_{i-1}=x_{\ge i-1}-(1-\delta)x_{\pi(m)}= x_{\ge i-1}-(1-\delta)(z_{\pi(m)}-y_{\pi(m)}) \\
        \le x_{\ge 1}-(1-\delta)(z_{\pi(m)}-\frac{y_{\ge 1}}{1-\delta}) = x_{\ge 1}+y_{\ge 1}-(1-\delta)z_{\pi(m)} \le 1-(1-\delta)\alpha,
    \end{multline}
    where the first inequality holds by Claim~\ref{cl:General FR Additiveness} (assigning $X=\max_{i\le<j\le n}v_j$, $Y=v_{i-1}$ and $T=Q(1-\delta)$); the first equality holds since in this case $\pi(m)=i-1$;  the second inequality holds since $x_{\ge 1}\ge x_{\ge j}$ for all $j\in[n]$, and since $\frac{y_{\ge 1}}{1-\delta}\ge y_{j}$ by Claim~\ref{cl:Semi Deterministic Monotonicity} (assigning $X=v_j$, $Y=\max_{k\neq j}v_k$ and $T=Q(1-\delta)$); the last inequality holds since $\pi(m)$ is a major-box.

    Thus, it holds from Inequalities~\eqref{eq:z bound} and \eqref{eq:x bound} that $z_{\ge i} < \frac32-\alpha+3\delta$ (since $2+\alpha<3$). 
    We can bound the RHS of Inequality~\eqref{eq:ratio-lem4} by its value at the point $x_{\geq i} = \frac12z_{\geq i}$\footnote{Although this point is infeasible, the value at this point can still be used as a lower bound on the ratio $\frac{\E[\ALG]}{\E[\OPT]}$.} and $z_{\ge i} = \frac32-\alpha+3\delta$ is $$\frac{1}{2.5-\alpha+3\delta}-16\delta.$$
    
    For $\delta<0.02$, $\alpha>\frac34$ and $\beta<\frac\alpha4$, it holds that $\frac{1}{2.5-\alpha+3\delta}>\frac{1}{2-\beta}$. Thus, the identity-truthfulness gap $\frac{\E[\ALG]}{\E[\OPT]}$ in this case is at least $\frac{1}{2-\beta}-16\delta$. This concludes the proof of this case.
    
    \paragraph{Case 2: $\E[\max_{j\ge i\wedge j \neq\pi(m)} v_j]\le \beta$.}
    If the algorithm reaches step $i$, where $\E[\max_{j\ge i\wedge j \neq\pi(m)} v_j] \le \beta$, before $v_{\pi(m)}$ has passed, the threshold is switched to $\sqrt{\alpha\beta}$. Let $p:=\Pr\left[\max_{j\ge i\wedge j \neq\pi(m)} v_j \ge \sqrt{\alpha\beta}\right]$ be the probability that any remaining box other than the major-box is higher than the threshold. By Markov's inequality, it holds that
    \begin{equation}\label{eq:p bound}
        p\le\frac{\E[\max_{j\ge i\wedge j \neq\pi(m)} v_j]}{\sqrt{\alpha\beta}}\le\frac{\beta}{\sqrt{\alpha\beta}}=\sqrt{\frac{\beta}{\alpha}}.
    \end{equation}
    We bound the algorithm's expectation using $p$
    \begin{equation*}
        \E[\ALG] \ge p\sqrt{\alpha\beta}+(1-p)\E\left[v_{\pi(m)}\cdot\ind{v_{\pi(m)}\ge\sqrt{\alpha\beta}}\right].
    \end{equation*}
    We can bound the major-box's contribution in the above inequality as
    \begin{align*}
        \E\left[v_{\pi(m)}\cdot\ind{v_{\pi(m)}\ge\sqrt{\alpha\beta}}\right]&=\E\left[v_{\pi(m)}\right]-\E\left[v_{\pi(m)}\cdot\ind{v_{\pi(m)}<\sqrt{\alpha\beta}}\right] \ge \alpha-\sqrt{\alpha\beta},
    \end{align*}
    to get the following final bound
    \begin{equation*}
        \E[\ALG] \ge p\sqrt{\alpha\beta}+(1-p)(\alpha-\sqrt{\alpha\beta})=\alpha-\sqrt{\alpha\beta}-p(\alpha-2\sqrt{\alpha\beta})\stackrel{\eqref{eq:p bound}}{\ge} \alpha-\sqrt{\alpha\beta}-\sqrt{\frac{\beta}{\alpha}}(\alpha-2\sqrt{\alpha\beta})=\alpha+2\beta-2\sqrt{\alpha\beta},
    \end{equation*}
    where the last inequality holds since $\alpha > 4\beta$ implies $\alpha-2\sqrt{\alpha\beta}>0$.
    Since the optimum online algorithm is trivially bounded by $\E[\OPT]\le 1$, we obtain
    $$
    \frac{\E[\ALG]}{\E[\OPT]}\ge \alpha+2\beta-2\sqrt{\alpha\beta}.
    $$

    Thus, the ratio of the expectations is the minimum of the ratios in the two cases, concluding the proof.
\end{proof}

Using the proposed algorithms, as well as the separation between $\delta$-threshold-hard distributions and non-$\delta$-threshold-hard distributions, we construct a simple decision tree to choose the most suitable algorithm. This allows us to fine-tune the parameters $\alpha$, $\beta$ and $\delta$ to minimize the gap between our upper bound on $\OPT$ and lower bound on $\ALG$.

\maintheorem*
\begin{proof}
    
    Consider Algorithm~\ref{alg:Identity-Truthful Algorithm} with parameters $(\alpha,\ \beta,\ \delta)= (0.854, 0.073, 0.001)$, we show that this algorithm guarantees $$\Gamma(\ALG)\ge 0.500001.$$
    \begin{algorithm}
        \caption{Identity-Truthful Algorithm}\label{alg:Identity-Truthful Algorithm}
        \begin{algorithmic}[1]
            \If {$\sup_p \T(p) \ge \frac12+\delta^2$}
                \State Use a single-threshold algorithm with $\tau=Q(p^*)$, where $p^*=\arg\sup_p \T(p)$. 
            \ElsIf {$\max_{i\in[n]}\E[v_i]\le\alpha$}
                \State Use Algorithm~\ref{alg:No Major Box Algorithm}.\Comment{No Major-Box Algorithm.}
            \Else
                \State Use Algorithm~\ref{alg:With Major Box Algorithm}.\Comment{Major-Box Algorithm.}
            \EndIf
        \end{algorithmic}
    \end{algorithm}
    
    By Claim~\ref{cl:Single Threshold Algorithm Bound}, by Definition~\ref{def:t}, and by Lemmas~\ref{lem:No Major Box Algorithm Bound} and \ref{lem:With Major Box Algorithm Bound}, we obtain that
    $$
    \Gamma(\ALG)=\min\Big(\underbrace{\frac12+\delta^2}_{\text{Claim~\ref{cl:Single Threshold Algorithm Bound}}},\ 
    \underbrace{\frac{2}{\alpha+3}-16\delta}_{\text{Lemma~\ref{lem:No Major Box Algorithm Bound}}},\ 
    \underbrace{\frac{1}{2-\beta}-16\delta,\ 
    \alpha+2\beta-2\sqrt{\alpha\beta}}_{\text{Lemma~\ref{lem:With Major Box Algorithm Bound}}}\Big).
    $$
    The values of $\alpha$, $\beta$ and $\delta$ were chosen through a numerical maximization process over this gap.
    These values satisfy the parameter ranges $\delta\in(0,0.02)$, $\alpha\in(\frac34+\frac52\delta,1)$, and $\beta\in(0,\frac\alpha4)$ (by Claim~\ref{cl:one major box}, $\alpha>\frac34+\frac52\delta$).
\end{proof}

\section{Identity-Truthfulness Gap Hardness}\label{sec:hardness}

In this section, we establish a hardness result of an upper bound on the identity-truthfulness gap for any randomized identity-truthful algorithm.
We show that there is an instance for which the maximal identity-truthfulness gap is strictly lower than the established hardness bound for the order-competitive ratio, 0.829 \citep{chen:2024-settingtargets}.

\thirdtheorem*

\begin{proof}
    Consider an instance consisting of $N+1$ boxes with the following rewards and probabilities:
    \begin{itemize}
        \item $N$ Bernoulli boxes, with a reward of $1$ with probability $\frac1N$, and $0$ otherwise.
        \item One large box, with a reward of $N$ with probability $\frac1N$, and $1$ otherwise.
    \end{itemize}
    For this instance, the online optimum always chooses the large box, since its arrival step is known, thus guaranteeing $2-\frac1N$. For an identity-truthful algorithm, it is always optimal to accept a reward of $N$ if observed and reject $0$. It further holds that the optimal identity-truthful algorithm accepts the first observed reward of $1$ after the large box has passed. The only non-trivial decision the algorithm faces is at what probability it should accept a reward of $1$, when observed before the large box has passed.

    We formulate a linear problem to find the best possible guarantee of the optimal randomized identity-truthful algorithm. Consider the event $A_i$ where the algorithm reaches step $i$ and the previous $i-1$ boxes were revealed to be Bernoulli boxes. Let
    $$\forall_{i\in[N]}: a_i=\Pr[A_i] \text{ and  } p_i=\Pr[\text{accept Box $i$} \mid A_i \wedge v_i =1] \mbox{ and } b_i = a_i\cdot p_i,$$
    meaning $b_i$ is the probability of accepting a value of $1$ conditioned on reaching step $i$, and that $v_i=1$.

    It holds that $0\le b_i\le a_i\le 1$. If the large box arrives after step $i$, then it holds that $a_{i+1}=a_i- \frac1N b_i$, since the algorithm reaches step $i+1$ if and only if (1) $A_i$ occurs (with probability $a_i$), (2) $v_i=1$ (which happens with probability $\frac1N$), and (3) the algorithm did not accept $v_i$ (which happens with probability $p_i$).

    The expected chosen reward of such an algorithm, under arrival order $\pi_k$ (parametrized by $k$) for which the large box is the $k$-th arriving box, is as follows
    \begin{gather*}
        \underbrace{1-a_k}_{\substack{\text{selecting $1$} \\ \text{before step } k}}+
        \underbrace{a_k\left(\frac1N\cdot N + \left(1-\frac1N\right)\cdot p_k\right)}_{\text{expected value from Box $k$}}+
        \underbrace{a_k \left(1-\frac1N\right)\left(1-p_k\right)}_{\substack{\text{probability of reaching $k+1$}} } \cdot\underbrace{\left(1-\left(1-\frac1N\right)^{N+1-k}\right)}_{\substack{\text{expected value from}  \\ \text{boxes after $k$}}} \\
            =1+a_k\left[1-\frac1N-\left(1-\frac1N\right)^{N+2-k}\right] + b_k\left(1-\frac1N\right)^{N+2-k}.
    \end{gather*}
    The following LP gives the maximal possible identity-truthfulness gap any randomized identity-truthful algorithm can achieve:
    \begin{align*}
        \text{maximize} \quad & \Gamma \\
        \text{subject to } \quad & 1+a_i\left[1-\frac1N-\left(1-\frac1N\right)^{N+2-i}\right] + b_i\left(1-\frac1N\right)^{N+2-i}\ge \Gamma\left(2-\frac1N\right), && \forall i \in \{1, \dots, N+1\} \\
        & a_{i+1}=a_i-\frac1N b_i, && \forall i \in \{1, \dots, N\} \\
        & 0\le b_i\le a_i\le1, && \forall i \in \{1, \dots, N+1\}
    \end{align*}

    For $N=114$, it holds that the solution for the LP is $\Gamma^*\approx0.8099 <0.81$ concluding the proof.
\end{proof}

\begin{remark}
    We note that as $N\to\infty$ the identity-truthfulness gap given by the LP approaches approximately $0.80826$.
\end{remark}

\section{Identity-Truthful Non-Monotonicity}\label{sec:pricing}
In this section, we show that under the constraint of identity-truthful mechanisms, there are instances for which the optimal algorithm does not adhere to a pricing mechanism.
In the single choice prophet inequality setting, a dynamic pricing mechanism is a sequence of thresholds $\tau_1,
\ldots,\tau_n$, where $\tau_i$ can depend on the information available to the online algorithm (like the identities of previously arriving boxes), such that at step $i$, $v_i$ is accepted if and only if $v_i\geq \tau_i$. 
Pricing mechanisms are of particular interest because they are common in real-world markets and enjoy strong strategyproofness guarantees. In particular, in a dynamic posted-price mechanism for selling an item, buyers have no incentive to misrepresent their values: each buyer’s optimal strategy is simply to accept the offered price if and only if their value exceeds it.

We present an instance for which any pricing mechanism has in expectation lower guarantees than the optimal online algorithm (with the same information). This holds both for worst-case ordering as well as random order (as in the prophet secretary setting).

\secondtheorem*

\begin{proof}
    Consider an instance consisting of two boxes with the following rewards and probabilities: 
    \[
    \begin{array}{l|cccc|c}
    \text{Reward} & 0 & 1 & 2 & 7 & \text{Expectation} \\ \hline
    \text{Box A} & 0.6 & 0.1 & 0.3 & - & 0.7 \\
    \text{Box B} & - & 0.7 & - & 0.3 & 2.8
    \end{array}
    \]
    It is always optimal to accept $7$ if observed and reject $0$. Moreover, it is always better to accept the value of the second box (if reached). The only non-trivial decision an online algorithm faces is whether to accept a value of $1$ or $2$ when observed on the first arriving box.
    The following table gives the expected value of each of the 4 possible (deterministic) strategies depending on whether to accept the values of $1$ and $2$ when observed in the first box for each of the two possible orders (``AB'' or ``BA'') alongside the worst case order, and the performance under a random order (where both  ``AB'' and ``BA'' are equally likely). 
    \[
    \begin{array}{c|c|c}
     & \text{Accept 1} & \text{Reject 1} \\ \hline
    \text{Accept 2} & 
    \begin{array}{l} 
    \text{``AB''} \colon\ 0.6 \cdot 2.8 + 0.7 = 2.38 \\ 
    \text{``BA''} \colon\ 2.8 \\ 
    \textbf{worst-case } \mathbf{2.38} \\ 
    \textbf{random } \mathbf{2.59} 
    \end{array} & 
    \begin{array}{l} 
    \text{``AB''} \colon\ 0.7 \cdot 2.8 + 0.3 \cdot 2 = 2.56 \\ 
    \text{``BA''} \colon\ 0.3 \cdot 7 + 0.7 \cdot 0.7 = 2.59 \\ 
    \textbf{worst-case } \mathbf{2.56} \\ 
    \textbf{random } \mathbf{2.575} 
    \end{array} \\ \hline
    \text{Reject 2} & 
    \begin{array}{l} 
    \text{``AB''} \colon\ 0.9 \cdot 2.8 + 0.1 \cdot 1 = 2.62 \\ 
    \text{``BA''} \colon\ 2.8 \\ 
    \textbf{worst-case } \mathbf{2.62} \\ 
    \textbf{random } \mathbf{2.71} 
    \end{array} & 
    \begin{array}{l} 
    \text{``AB''} \colon\ 2.8 \\ 
    \text{``BA''} \colon\ 0.3 \cdot 7 + 0.7 \cdot 0.7 = 2.59 \\ 
    \textbf{worst-case } \mathbf{2.59} \\ 
    \textbf{random } \mathbf{2.695} 
    \end{array}
    \end{array}
    \]
    The table above implies that the optimal deterministic identity-truthful algorithm is to accept a value of $1$ and reject a value of $2$ (both for worst-case and for random order).
    This algorithm is not a pricing mechanism (in contrast to the three other analyzed algorithms). We note that even for randomized identity-truthful algorithms, the optimal pricing mechanism obtains $2.66$ (accept 1 \& 2 w.p. $\frac13$, otherwise reject both), while the optimal identity-truthful algorithm obtains approximately $2.7$ (always reject 2; accept 1 w.p. $\frac7{13}$).
\end{proof}

We note that the optimum online algorithm, as well as the optimum order-unaware non-identity-truthful algorithm (i.e., that can use the identity of Box $i$ for the decision whether to accept Box $i$) 
is to always accept box B, thus guaranteeing an expected reward of $2.8$. Thus, for this instance, it holds that $$\text{non-identity-truthful}>\text{identity-truthful}>\text{identity-truthful pricing mechanism}.$$

We also note that for every instance, both the optimum order-aware online algorithm and the optimum order-unaware non-identity-truthful algorithm can always be implemented by a dynamic pricing mechanism. Theorem~\ref{thm:Non-Monotone} states that this is not the case for identity-truthful mechanisms.

Our paper raises two major open questions, namely, whether the optimal identity-truthfulness gap can be achieved through a pricing mechanism, and whether the identity-truthfulness gap is different than the order-competitive ratio. 

\section{Additional Implications of our Techniques}\label{sec:single-threshold}
In this section, we show how our lower bound $\T(p)$  (Definition~\ref{def:t}) can be utilized to show alternative proofs for obtaining a competitive-ratio of $\frac{1}{2}$ for several known single-threshold algorithms. Namely, we prove (Claims~
\ref{cl:median threshold} and \ref{cl:t2} respectively) that $\T(\frac12)\ge\frac12$ (median threshold \cite{ester:1984}) and $\T(Q^{-1}(\frac12)) \ge \frac12$ (half-mean threshold \cite{kleinberg2012matroid}). Claim~\ref{cl:Area Claim} serves as an alternative proof for the competitiveness of using a random prophet sample as the threshold \citep{rubinstein:2020}. We additionally prove (in Claim~\ref{cl:three-quarters threshold}) that $\T(p)\ge\frac12$ for a newly discovered threshold that satisfies $\int_p^1Q(q)dq=\frac34$ \citep{brustle:2026}.
Furthermore, we show (in Claim~\ref{cl:high reward interval}) that the interval of quantiles for which their corresponding thresholds obtain a competitive-ratio of at $\frac12$ is of length at least $\frac12$.
This provides additional insights regarding the interval of thresholds that are $\frac12$-competitive, which is an open question from \citep{waggoner2018prophet}.
A recent independent note by \citet{zhang2026thresholdrulesclassicalprophet} also gives a unified threshold/surplus perspective on several classical single-threshold prophet inequality rules, including the median threshold, the half-mean threshold, and the sample-max threshold.
Our approach is different: we phrase the analysis in terms of the quantile function of the prophet and the function $\T(p)$, which also yields structural information about the set of quantiles whose corresponding thresholds are $\frac12$-competitive.

\begin{claim}\label{cl:median threshold}
    It holds that $\T(\frac12)\ge\frac12$.
\end{claim}
\begin{proof}
    \begin{align*}
        \T(\frac{1}{2})=\frac{1}{2}\int_{\frac{1}{2}}^1 \frac{Q(q)}{q^2}dq 
        &=\frac{1}{2}\int_{\frac{1}{2}}^1 Q(q)\left[\frac{1}{q^2}-1\right]dq+\frac{1}{2}\int_{\frac{1}{2}}^1 Q(q)dq \\
        &\ge\frac{Q(\frac{1}{2})}{2}\int_{\frac{1}{2}}^1\left[\frac{1}{q^2}-1\right]dq+\frac{1}{2}\int_{\frac{1}{2}}^1 Q(q)dq \\
        &=\frac12\left(\frac{Q(\frac{1}{2})}{2}+\int_{\frac{1}{2}}^1 Q(q)dq\right) \\
        &\ge\frac{1}{2}\int_0^1 Q(q)dq=\frac12,
    \end{align*}
    where the third equality holds since $\int_{\frac{1}{2}}^1\left[\frac{1}{q^2}-1\right]dq=\frac12$, and the first and second inequalities hold since $Q$ is non-decreasing. For example, the second inequality holds since $\frac12Q(\frac12)\ge\int_0^{\frac12}Q(q)dq$.
\end{proof}

\begin{claim} \label{cl:t2}
    For $p^*$ such that $Q(p^*)=\frac12$ it holds that $\T(p^*)\ge\frac12$.
\end{claim}
\begin{proof}
    We first notice that
    \begin{equation}\label{eq:Q integral}
        2Q(p^*)=1=\int_0^1Q(q)dq\le p^*Q(p^*)+\int_{p^*}^1Q(q)dq\implies\int_{p^*}^1Q(q)dq\ge(2-p^*)Q(p^*).
    \end{equation}
    Therefore, we have that
    \begin{align*}
        \T(p^*)=p^*\int_{p^*}^1 \frac{Q(q)}{q^2}dq 
        &=p^*\int_{p^*}^1 \frac{Q(q)-Q(p^*)}{q^2}dq+p^*\int_{p^*}^1 \frac{Q(p^*)}{q^2}dq \\
        &\ge p^*\int_{p^*}^1\left(Q(q)-Q(p^*)\right)dq+(1-p^*)Q(p^*) \\
        &= p^*\int_{p^*}^1Q(q)dq-p^*(1-p^*)Q(p^*)+(1-p^*)Q(p^*) \\
        &\ge p^*(2-p^*)Q(p^*)+(1-p^*)^2Q(p^*) \\
        &= Q(p^*),
    \end{align*}
    where the first inequality holds since $\frac1{q^2}\ge1$ and $x\int_{x}^1 \frac{1}{t^2}dt=1-x$, the second inequality holds by Inequality~\eqref{eq:Q integral} and the last equality holds since $x(2-x)+(1-x)^2=1$. By the definition of $p^*$, it holds that $\T(p^*)\ge Q(p^*)=\frac12$, which concludes the proof.
\end{proof}

\citet{brustle:2026} recently showed a new $\frac12$-competitive threshold $t$ that satisfies $\E[Z\cdot\ind{Z\ge t}] = \frac34$. We next show how to derive this result using our function $\T$.
\begin{claim} \label{cl:three-quarters threshold}
   For $p^*$ such that $\int_{p^*}^1Q(q)dq=\frac34$ it holds that $\T(p^*)\ge\frac12$.
\end{claim}
\begin{proof}
    We first notice that
    \begin{equation}\label{eq:Q integral 3/4}
        1=\int_0^1Q(q)dq\le p^*Q(p^*)+\int_{p^*}^1Q(q)dq=p^*Q(p^*)+\frac34\implies Q(p^*)\ge\frac{1}{4p^*}.
    \end{equation}
    Therefore, we have that
    \begin{align*}
        \T(p^*)=p^*\int_{p^*}^1 \frac{Q(q)}{q^2}dq 
        &=p^*\int_{p^*}^1 \frac{Q(q)-Q(p^*)}{q^2}dq+p^*\int_{p^*}^1 \frac{Q(p^*)}{q^2}dq \\
        &\ge p^*\int_{p^*}^1\left(Q(q)-Q(p^*)\right)dq+(1-p^*)Q(p^*) \\
        &= p^*\int_{p^*}^1Q(q)dq-p^*(1-p^*)Q(p^*)+(1-p^*)Q(p^*) \\
        &= \frac34p^*+(1-p^*)^2Q(p^*) \\
        &\ge \frac34p^*+\frac{(1-p^*)^2}{4p^*},
    \end{align*}
    where the first inequality holds since $\frac1{q^2}\ge1$ and $x\int_{x}^1 \frac{1}{t^2}dt=1-x$, the last equality holds by the definition of $p^*$ and the last inequality holds by Inequality~\eqref{eq:Q integral 3/4}. In the range $p^*\in[0,1]$, the expression $\frac34p^*+\frac{(1-p^*)^2}{4p^*} \ge \frac12$, concluding the proof.
\end{proof}

\begin{claim}\label{cl:high reward interval}
    The interval of quantile thresholds that are $\frac12$-competitive is at least of length $\frac12$.
\end{claim}
\begin{proof}
   Let $a= \min \{ x\geq 0 \mid  \T(x) \geq \frac12\}$, and let  $b= \max \{ x\leq 1 \mid  \T(x) \geq \frac12\}$. $a$ and $b$ are well defined since there must exist a quantile threshold (e.g., the median) for which $\T(x)\geq \frac12$ and since we consider atomless distributions (which is without loss of generality).
   Since $\T$ is concave (Claim~\ref{cl:Concavity Claim}), then $[a,b]$ is the interval of $\frac{1}{2}$-competitive quantile thresholds.  We need to show that  $b-a\ge\frac12$. It holds that \begin{equation}
       a\left(\T(a)-\frac12\right)= 0,\label{eq:ata}
   \end{equation}
   since if $\T(0)\ge\frac12$ then $a=0$ and otherwise ($\T(0)<\frac12$) then $a>0$ which implies that $\T(a)=\frac12$ (by continuity of $\T$). Thus, either $a=0$ or $\T(a)=\frac12$. Lastly, it holds that $\T(b)=\frac12$ by continuity of $\T$ and since $\T(1)=0$.

    Consider the area enclosed under the two tangent lines at $p=a$ and $p=b$. By concavity and by Claim~\ref{cl:Area Claim}, this area must satisfy
    \begin{equation*}
        \frac12=\int_0^1\T(p)dp\le\int_0^1\min\{\T(a)+(p-a)\T'(a), \T(b)+(p-b)\T'(b)\}dp.
    \end{equation*}
    By subtracting half from both sides, we have
    \begin{align}
        0&\le \int_0^1\min\{\T(a)-\frac12+(p-a)\T'(a), (p-b)\T'(b)\}dp \nonumber\\
        &\le \int_0^{2a}\left(\T(a)-\frac12+(p-a)\T'(a)\right)dp + \int_{2a}^1(p-b)\T'(b)dp \nonumber\\
        &= 2a\left(\T(a)-\frac12\right)+\T'(b)\left(\frac12 -b-2a^2+2ab\right)  \nonumber\\
        &=\frac12\T'(b)(1-2a)(1+2a-2b), \label{eq:ab}
    \end{align}
    where the last equality holds by Equation~\eqref{eq:ata}.
    We note that $a\le \frac12$ since the median is $\frac12$-competitive, and $\T'(b)\le 0$ since $\T(b+\epsilon) < \T(b) $ for every $\epsilon>0$. Therefore, Inequality~\eqref{eq:ab} implies that either $a=\frac12$, or $\T'(b)=0$, or $2a-2b+1\le 0$ (which implies $b-a\ge \frac12$).

   When $a=\frac12$, consider only the tangent line at $p=a$, of $\T(a)+(p-a)\T'(a)$. For this case, $\T(a)=\frac12$, and the area under the tangent $\int_0^1\left(\frac12+(p-\frac12)\T'(a)\right)dp=\frac12$. By Claim~\ref{cl:Area Claim}, the area under $\T$ equals $\frac12$, therefore $\T(p)$, which is bounded by the tangent line, must equal the tangent line. This is a straight line passing through $\T(\frac12)=\frac12$, therefore, the interval of quantiles above $\frac12$ is of length at least $\frac12$.

   When $\T'(b)=0$, consider only the tangent line at $p=b$, of $\T(b)+(p-b)\T'(b)$. For this case, the tangent line is a constant at $\frac12$. The area under the tangent is $\frac12$, then, similarly, it must hold that $\T(p)=\frac12$ for every $p$, in contradiction with $\T(1)=0$.

   Finally, if $2a-2b+1\le 0$  then  $b-a\ge\frac12$. Therefore, under all possible cases, the interval of quantile thresholds that are $\frac{1}2$-competitive is of length at least $\frac12$, concluding the proof.
\end{proof}

We next prove using our function $\T$ that the threshold that is equal to the surplus above it (analyzed by \citet{ester:1984}) obtains a competitive-ratio of $\frac12$. 
\begin{claim}\label{cl:balanced-surplus}
   For $p^*$ satisfying  $Q(p^*)=\sum_{i=1}^n \mathbb{E}[(X_i-Q(p^*))^+]$ it holds that $\T(p^*)\ge\frac12$.
\end{claim}
\begin{proof}
    Denoting $F_{\max}=\prod_{i\in[n]} F_i$, we can give the following bound
    \begin{align}\label{eq:surplus prob}
    	\forall_{x\ge 0}:\sum_{i=1}^n \Pr[X_i>x]&=\sum_{i=1}^{n}(1-F_i(x)) \nonumber\\
    	&= \sum_{i=1}^{n}(1-F_i(x))\frac{\prod_{j<i}F_j(x)}{\prod_{j<i}F_j(x)} \nonumber\\
    	&\le \frac{1}{F_{\max}(x)}\sum_{i=1}^{n}(1-F_i(x))\left(\prod_{j<i}F_j(x)\right) \nonumber\\
    	&= \frac{1-F_{\max}(x)}{F_{\max}(x)}.
    \end{align}
    
    We can bound the sum surplus as follows
    \begin{align*}
    	\sum_{i=1}^n \mathbb{E}[(X_i-Q(p))^+]&=\int_{Q(p)}^{\infty}\sum_{i=1}^n \Pr[X_i>x]dx \nonumber\\
        &\stackrel{\eqref{eq:surplus prob}}{\le} \int_{Q(p)}^{\infty} \frac{1-F_{\max}(x)}{F_{\max}(x)}dx \nonumber\\
    	&= \int_p^1 \left(\frac{1}{q}-1\right)Q'(q)dq \nonumber\\
    	&= \int_p^1 \frac{Q(q)}{q^2}dq+\left[Q(q)(\frac{1}{q}-1)\right]_{p}^1 \nonumber\\
    	&= \int_p^1 \frac{Q(q)}{q^2}dq+Q(p)\left(1-\frac1p\right) \nonumber\\
    	&= \int_p^1 \frac{Q(q)}{q^2}dq+Q(p)\int_p^1\frac{-1}{q^2}dq \nonumber\\
    	&= \int_p^1 \frac{Q(q)-Q(p)}{q^2}dq.
    \end{align*}
    where the second equality is given by the substitution $x=Q(q)$ and the third equality is given by applying integration by parts using $u=\frac{1}{q}-1$ and $v=Q(q)$.
    
    Therefore, it holds that $Q(p^*)\le \int_{p^*}^1 \frac{Q(q)-Q(p^*)}{q^2}dq$, and we can denote a new quantile threshold $p_{\max}\ge p^*$ such that $Q(p_{\max})= \int_{p_{\max}}^1 \frac{Q(q)-Q(p_{\max})}{q^2}dq$, for which
    \begin{align*}
    	Q(p_{\max})&=\int_{p_{\max}}^1 \frac{Q(q)-Q(p_{\max})}{q^2}dq=\int_{p_{\max}}^1 \frac{Q(q)}{q^2}dq+Q(p_{\max})\left(1-\frac1{p_{\max}}\right) \\
        &\implies Q(p_{\max})=p_{\max}\int_{p_{\max}}^1 \frac{Q(q)}{q^2}dq=\T(p_{\max}).
    \end{align*}
    
    This inequality implies $\T'(p_{\max})=\frac1{p_{\max}}(\T(p_{\max})-Q(p_{\max}))=0,$ meaning that this threshold achieves the maximal value of $\T(p)$ (which is concave by Claim~\ref{cl:Concavity Claim}) and $\T(p_{\max})\ge\frac12$. It is known that the balanced-surplus threshold is higher than half of the prophet expectation \citep{waggoner2018prophet}, therefore $p^*\in[Q^{-1}(\frac12),\ p_{\max}]$, where $\T(Q^{-1}(\frac12))\ge\frac12$ (Claim~\ref{cl:t2}). By concavity $\T(p^*)\ge\frac12$, concluding the proof.
\end{proof}


Claim~\ref{cl:balanced-surplus} shows that the balanced-surplus threshold is $\frac12$-competitive but also that there is a higher balance-based threshold, that satisfies $Q(p_{\max})=\T(p_{\max})$, that is also $\frac12$-competitive and is optimal (obtaining the maximum guarantee for all instances with the CDF of the prophet). This threshold has been independently found by \citet{correa:2026} (which is concurrent to this paper). This (together with Claim~\ref{cl:allQ}) answers the question posed by \citet{waggoner2018prophet} on what is the interval of $\frac12$-competitive thresholds. Terming the threshold $Q(p_{\max})=\T(p_{\max})$ ``balanced-guarantee'', we have an ordering
$$\text{half-mean}\stackrel{}{\leq}\text{balanced-surplus}\leq\text{balanced-guarantee},$$
where the (worst-case) expected reward of each threshold increases as the threshold increases in this range, and decreases beyond this range.
The median threshold and the threshold in Claim~\ref{cl:three-quarters threshold} do not admit a consistent ordering.
Lastly, the random prophet-sample has the lowest worst-case guarantee (i.e., when the instance is composed of increasing weighted Bernoulli boxes), which equals $\frac12$ regardless of the specific prophet distribution.

 \bibliographystyle{abbrvnat}
 \bibliography{bib}

\appendix

\section{Adapting Order-Unaware Algorithms}\label{apx:order-unaware}
In this section, we consider the  order-unaware algorithms suggested by \cite{ezra:2023-whoisnext} and \cite{chen:2024-settingtargets}.
We adapt these algorithms such that they do not utilize the identity of the arriving box, thus converting them into identity-truthful algorithms.
Denote by $\tau_i$ the threshold at step $i$ used by an order-unaware algorithm. If we construct a different algorithm where at step $i$ the threshold $\tau_{i-1}$ is used, we create an identity-truthful algorithm since the identity of the arriving box at step $i$ is not known when the threshold $\tau_{i-1}$ was calculated.
Thus, we have the following two identity-truthful algorithms, denoted by the two sequences of thresholds $\tau^1_i$ and $\tau^2_i$ (by \cite{ezra:2023-whoisnext} and \cite{chen:2024-settingtargets} resp.):
$$
\tau^1_i = \max\left\{\alpha_i, \beta_i\right\} \quad \text{where}\quad
\alpha_i=\frac{1}{\phi}\E[\max_{j\ge i} v_j] \quad
\beta_i=x \text{ satisfying } \E\left[(\max_{j\ge i} v_j-\phi\cdot x)^+\right]=x,
$$
$$
\tau^2_i = \min\left\{x\ge 0 \mid \E[\max\{v_{i-1},x\}]\ge \tau^2_{i-1} \right\} \quad  \text{where}\quad
\tau^2_0=\tau^2_1=\frac{1}{\phi}\E[\max_{i\in[n]} v_i].
$$

Consider the following instance \textendash\ the first box with reward $1$ with probability $\frac{1}{2}$ and $0$ otherwise, and the second box with reward $\frac{1}{\epsilon}$ with probability $\epsilon$ and $1$ otherwise.
It can be checked that the thresholds used for this instance by the two algorithms are $\tau^1_1=\tau^1_2=\tau^2_1=\tau^2_2=\frac2\phi>1$.
Consequently, a value of $1$ is never chosen. Thus both algorithms only take $\frac1\epsilon$ with probability $\epsilon$, with an expectation of $1$ while the optimal algorithm takes the second box with an expectation of $2-\epsilon$.
Therefore, both algorithms approach a ratio of $\frac12$.

\section{Missing  Claims}\label{apx:technical-claims}

\begin{claim}\label{cl:Semi Deterministic Monotonicity}
    Let $X$ and $Y$ be independent random variables. Let $Z=\max\{X,Y\}$ and $T\ge0$. Then,
    $$\E[Z\cdot \ind{Z < T}] \ge \E[X\cdot\ind{X<T}]\Pr[Z<T].$$
\end{claim}
\begin{proof}
    $\E[Z\cdot \ind{Z< T}]\ge \E[X\cdot \ind{Z< T}]=\E[X\cdot\ind{X<T}]\Pr[Y<T]\ge\E[X\cdot\ind{X<T}]\Pr[Z<T]$.
\end{proof}

\begin{claim}\label{cl:General FR Additiveness}
    Let $X$ and $Y$ be non-negative independent random variables. Let $Z=\max\{X,Y\}$ and $T\ge0$. Then,
    $$\E[Z\cdot \ind{Z \ge T}]-\E[Y\cdot \ind{Y \ge T}]\le \E[X \cdot \ind{X\ge T}]\le\frac{\E[Z\cdot \ind{Z \ge T}]-\E[Y\cdot \ind{Y \ge T}]}{\Pr[Z<T]}.$$
\end{claim}
\begin{proof}
    For the lower bound, assume without loss of generality that $X\ge Y$, therefore $Z=X$ implying that $Z\cdot \ind{Z\ge T}= X\cdot \ind{X\ge T}$. Since $Y$ is non-negative, it holds that $Z\cdot \ind{Z\ge T}\le X\cdot \ind{X\ge T}+Y\cdot \ind{Y\ge T}$, concluding this bound.

    For the upper bound, we have 
    \begin{align*}
        \E[Z\cdot\ind{Z\ge T}]&=\E[Z\cdot\ind{X\ge T>Y}]+\E[Z\cdot\ind{Y\ge T}] \\
        &\ge \E[X\cdot\ind{X\ge T}]\Pr[Y<T]+\E[Y\cdot\ind{Y\ge T}] \\
        &\ge \E[X\cdot\ind{X\ge T}]\Pr[Z<T]+\E[Y\cdot\ind{Y\ge T}],
    \end{align*}
    which is given by splitting the event $Z\ge T$ by $Y<T$ and $Y\ge T$. For the event $Y<T$, since $Z\ge T$, it must be that $X\ge T$. The desired result is achieved after rearranging.
\end{proof}

\begin{claim}\label{cl:General Mid-Range Integral}
    Let $X$ and $Y$ be independent random variables. Let $Z=\max\{X,Y\}$ and $a<b$. Then,
    $$\E[(X-a)\cdot \ind{a\le X<b}]\le \frac{\E[(Z-a)\cdot \ind{a\le Z<b}]}{\Pr[Z<b]}.$$
\end{claim}
\begin{proof}
    \begin{align*}
        \E[(X-a)\cdot \ind{a\le X<b}]&=\int_{a}^{b} \left(\Pr[X<b]-\Pr[X<x]\right)dx \\
        &=\int_{a}^{b} \left(\frac{\Pr[Z<b]}{\Pr[Y<b]}-\frac{\Pr[Z<x]}{\Pr[Y<x]}\right)dx \\
        &\le\frac{1}{\Pr[Y<b]}\int_a^b \left(\Pr[Z<b]-\Pr[Z<x]\right)dx \\
        &\le\frac{1}{\Pr[Z<b]}\int_a^b \left(\Pr[Z<b]-\Pr[Z<x]\right)dx \\
        &= \frac{1}{\Pr[Z<b]}\E[(Z-a)\cdot \ind{a\le Z<b}],
    \end{align*}
    where the first inequality holds since $\Pr[Y<x]\le\Pr[Y<b]$ and the second inequality holds since $\Pr[Y<b]\ge \Pr[Z<b]$.
\end{proof}

\begin{claim} \label{cl:Function Optimization}
    Given $0\le L<U\le 1$, the function $$f(x,z)=\frac{\frac12+\left[\frac{z}{2}-x\right]^+}{\frac12-x+\max\{z,\frac12\}}$$ in the domain $x\in[0,z]$ and $z\in[L,U]$, has a global minimum at $(x,z)=(\frac12 L,L)$ or $(x,z)=(\frac12 U,U)$.
\end{claim}
\begin{proof}
    There are 4 regions:
    \begin{enumerate}
        \item $\{(x,z)~|~x\ge\frac12 z \wedge z\le\frac12\}$ \textendash\ $f(x,z)=f_1(x,z)=\frac{1}{2-2x}$.
        \item $\{(x,z)~|~x\le\frac12 z \wedge z\le\frac12\}$ \textendash\ $f(x,z)=f_2(x,z)=\frac{1+z-2x}{2-2x}$.
        \item $\{(x,z)~|~x\ge\frac12 z \wedge z\ge\frac12\}$ \textendash\ $f(x,z)=f_3(x,z)=\frac{1}{1+2z-2x}$.
        \item $\{(x,z)~|~x\le\frac12 z \wedge z\ge\frac12\}$ \textendash\ $f(x,z)=f_4(x,z)=\frac{1+z-2x}{1+2z-2x}$.
    \end{enumerate}
    
    We first show the derivatives with respect to $x$:
    \begin{enumerate}
        \item $\frac{\partial f_1}{\partial x} = \frac{2}{(2-2x)^2} \ge 0$.
        \item $\frac{\partial f_2}{\partial x} = \frac{2z-2}{(2-2x)^2} \le 0$.
        \item $\frac{\partial f_3}{\partial x} = \frac{2}{(1+2z-2x)^2}\ge 0$.
        \item $\frac{\partial f_4}{\partial x} = \frac{-2z}{(1+2z-2x)^2}\le 0$.
    \end{enumerate}
    We have that if $x\le\frac12 z$ then $\frac{\partial f}{\partial x} \le 0$, and if $x\ge\frac12 z$ then $\frac{\partial f}{\partial x} \ge 0$.
    This implies that the minimum is at the boundary between the regions, where $x=\frac12 z$.

    We next show the derivatives with respect to $z$:
    \begin{enumerate}
        \item $\frac{\partial f_1}{\partial z} = 0$.
        \item $\frac{\partial f_2}{\partial z} = \frac{1}{2-2x}\ge 0$.
        \item $\frac{\partial f_3}{\partial z} = \frac{-2}{(1+2z-2x)^2}\le 0$.
        \item $\frac{\partial f_4}{\partial z} = \frac{2x-1}{(1+2z-2x)^2}\le 0$ (since in this region $x\le\frac12z\le\frac12$).
    \end{enumerate}
    We have that if $z\le\frac12$ then $\frac{\partial f}{\partial z} \ge 0$, and if $z\ge\frac12$ then $\frac{\partial f}{\partial z} \le 0$.
    This implies that the minimum is achieved at boundary $z=L$ or $z=U$.
    Therefore, the global minimum is  achieved at $(x,z)=(\frac12 L,L)$ or $(x,z)=(\frac12 U,U)$.
\end{proof}

\begin{claim}\label{cl:one major box}
    Given parameters $\delta\in(0,0.25)$ and $\alpha$ such that $\alpha>\frac34+\frac52\delta$, if the prophet's distribution is $\delta$-threshold-hard, then at most one box $i$ satisfies $\E[v_i]\ge \alpha$.
\end{claim}
\begin{proof}
    If $\alpha>1$, then no box can have an expected value of at least $\alpha$, since the prophet's expectation (which is an upper bound on the expectation of every box) is normalized to $1$.  Otherwise, assume towards contradiction that there are two boxes $i\neq j$ such that $\E[v_i], \E[v_j]\ge\alpha$.
    By Claim~\ref{cl:Semi Deterministic Monotonicity} (Assigning $X=v_i$, $Y=\max_{k\neq i}v_k$ and $T=Q(1-\delta)$), it holds that
    \begin{align}\label{eq:major Det bound}
        y_i\le\frac{y_{\ge 1}}{1-\delta}\le\frac{\frac12+\delta}{1-\delta}.
    \end{align}
    Similarly, it holds that $y_j \le \frac{\frac12+\delta}{1-\delta}$.
    Therefore,
    \begin{align*}
        x_{\ge 1}&\ge\E\left[\max\{v_i,v_j\}\cdot\ind{\max\{v_i,v_j\}\ge Q(1-\delta)}\right] \\
        &\ge \Pr\left[\max\{v_i,v_j\}< Q(1-\delta)\right]x_i+x_j \\
        &\ge (1-\delta)x_i+x_j \\
        &= (1-\delta)(z_i-y_i)+z_j-y_j \\
        &\ge (2-\delta)\alpha-(\frac12+\delta)-\frac{\frac12+\delta}{1-\delta} \\
        &\ge 2\alpha-1-4\delta,
    \end{align*}
    where the first inequality holds since we take the maximum of a subset of boxes, the second inequality holds by Claim~\ref{cl:General FR Additiveness} (Assigning $X=v_i$, $Y=v_j$ and $T=Q(1-\delta)$), the fourth inequality holds since $i$ and $j$ are major-boxes and by Inequality~\eqref{eq:major Det bound}, and the last inequality holds since $\alpha\le1$ and $\delta<\frac14$.
    Using Claim~\ref{cl:FR Upper Bound} we have
    $$\frac12+\delta\ge x_{\ge 1}\ge 2\alpha-1-4\delta\implies\alpha\le\frac34+\frac52\delta,$$
    contradicting the condition that $\alpha>\frac34+\frac52\delta$.
\end{proof}

\end{document}